\documentclass[12pt,draftclsnofoot,onecolumn]{IEEEtran}
\usepackage{amsfonts}
\usepackage{cite,graphicx,amsmath,amsthm}
\usepackage{subfigure}
\usepackage{fancyhdr}
\usepackage{dsfont}
\usepackage{array,color}
\usepackage{bm}
\usepackage{float}
\usepackage{algorithm}
\usepackage{algpseudocode}
\usepackage{multirow}
\usepackage{bbm}
\usepackage{graphicx}
\usepackage{float}
\usepackage{subfigure}
\usepackage{ragged2e}
\usepackage{booktabs}
\usepackage{bbding}
\usepackage{multirow}
\usepackage{makecell}
\usepackage{amsmath,amssymb,amsfonts}
\usepackage[]{caption2}
\usepackage[table,xcdraw]{xcolor}
\usepackage{cite}

\newtheorem{theorem}{Theorem}

\newtheorem{lemma}{Lemma}

\newtheorem{proposition}{Proposition}

\newtheorem{corollary}{Corollary}

\newtheorem{property}{Property}

\newtheorem{remark}{Remark}

\newtheorem{claim}{Claim}

\newtheorem{assumption}{Assumption}

\allowdisplaybreaks[4]

\begin{document}
\title{{Edge Federated Learning Via Unit-Modulus Over-The-Air Computation}}

\author{
Shuai~Wang,~\IEEEmembership{Member,~IEEE},
Yuncong~Hong,~\IEEEmembership{Student~Member,~IEEE},
Rui~Wang,~\IEEEmembership{Member,~IEEE},
Qi Hao,~\IEEEmembership{Member,~IEEE},
Yik-Chung~Wu,~\IEEEmembership{Senior~Member,~IEEE},
and Derrick~Wing~Kwan~Ng,~\IEEEmembership{Fellow,~IEEE}

\thanks{
\scriptsize
Part of this paper has been presented at the IEEE Global Communications Conference (GLOBECOM), Madrid, Spain, Dec. 2021 \cite{globecom}.

Shuai Wang is with the Shenzhen Institute of Advanced Technology (SIAT), Chinese Academy of Sciences, Shenzhen 518055, China (e-mail: s.wang@siat.ac.cn).

Yuncong Hong and Rui Wang are with the Department of Electrical and Electronic Engineering, Southern University of Science and Technology (SUSTech), Shenzhen 518055, China (e-mail: hongyc@mail.sustech.edu.cn; wang.r@sustech.edu.cn).

Qi Hao is with the Department of Computer Science and Engineering and the Sifakis Research Institute of Trustworthy Autonomous Systems, Southern University of Science and Technology (SUSTech), Shenzhen 518055, China (e-mail: hao.q@sustech.edu.cn).

Yik-Chung~Wu is with the Department of Electrical and Electronic Engineering, The University of Hong Kong, Hong Kong (e-mail: ycwu@eee.hku.hk).

Derrick~Wing~Kwan~Ng is with the School of Electrical Engineering and Telecommunications, the University of New South Wales, Australia (email: w.k.ng@unsw.edu.au).

}
}

\maketitle

\vspace{-0.4in}
\begin{abstract}
Edge federated learning (FL) is an emerging paradigm that trains a global parametric model from distributed datasets based on wireless communications.
This paper proposes a unit-modulus over-the-air computation (UMAirComp) framework to facilitate efficient edge federated learning, which simultaneously uploads local model parameters and updates global model parameters via analog beamforming.
The proposed framework avoids sophisticated baseband signal processing, leading to low communication delays and implementation costs.
Training loss bounds of UMAirComp FL systems are derived and two low-complexity large-scale optimization algorithms, termed penalty alternating minimization (PAM) and accelerated gradient projection (AGP), are proposed to minimize the nonconvex nonsmooth loss bound.
Simulation results show that the proposed UMAirComp framework with PAM algorithm achieves a smaller mean square error of model parameters' estimation, training loss, and test error compared with other benchmark schemes.
Moreover, the proposed UMAirComp framework with AGP algorithm achieves satisfactory performance while reduces the computational complexity by orders of magnitude compared with existing optimization algorithms.
Finally, we demonstrate the implementation of UMAirComp in a vehicle-to-everything autonomous driving simulation platform.
It is found that autonomous driving tasks are more sensitive to model parameter errors than other tasks since the neural networks for autonomous driving contain sparser model parameters.
\end{abstract}

\IEEEpeerreviewmaketitle

\section{Introduction}

Deep learning has achieved unprecedented breakthrough in image classification, speech recognition, and object detection due to its ability to efficiently extract intricate nonlinear features from high-dimensional data \cite{dl1}.
Typically, a cloud center collects data from distributed users and trains a centralized model via gradient-based back propagation \cite{edge1,edge2,edge3,edge4}.
However, since the users need to share their local data to the cloud, this paradigm could lead to some potential privacy issues, hindering the development of deep learning in extensive applications such as smart cities and financial systems.

To address the privacy issue, federated learning (FL), which trains individual deep learning models at users, has been proposed by Google Research \cite{google}.
In the framework of FL, the locally generated data is locally adopted and not shared to any third party.
To leverage the knowledge from other users, local model parameters are uploaded periodically to a parameter server for model aggregation and the aggregated global parameters are broadcast to the users for further local updates.
Therefore, FL achieves distributed training while ensuring data privacy \cite{fedreview}.

\subsection{Edge Federated Learning and Related Work}

FL was originally developed for wire-line connected systems \cite{fed1}.
To achieve ubiquitous intelligence, a promising solution is edge FL, e.g., \cite{fedreview,fed1,fed2,fed3,fed3b,fed4,fed5,air2,air3,air4,air5}, where users are connected to an edge server via wireless links.
However, the convergence of edge FL may take a long time due to limited capacity of wireless channels during the uplink model aggregation step.
To reduce the transmission delay, various edge FL designs have been proposed (summarized in Table I\footnote{For more related work on digital and analog federated learning, please refer to \cite{fedreview}.}), which are mainly categorized into digital modulation \cite{fed1,fed2,fed3,fed3b,fed4,fed5} and analog modulation \cite{air2,air3,air4,air5} methods.

For digital modulation and single-antenna systems, data from different users are multiplexed either in the time or the frequency domain.
Current works on delay reduction focus on reducing 1) the number of model aggregation iterations \cite{fed2}, 2) the number of users \cite{fed3}, or 3) the number of bits for representing the gradient of back propagation in each iteration \cite{fed4}.
However, since these strategies involve approximation or simplification of the FL procedure, the performance of learning would be degraded inevitably.
Another way to reduce the transmission delay is to adopt multiple-input multiple-output (MIMO) technology for transmission so that data from multiple users are multiplexed concurrently in the spatial domain \cite{fed5}.

\begin{table*}[!t]
\caption{A Comparison of Existing and Proposed Schemes.}
\vspace{0.1in}
{

\centering
\begin{tabular}{|c|c|c|c|c|c|c|c|c|}
\hline
\hline
\textbf{Modulation} & \textbf{Work} & \textbf{MIMO} & \textbf{\begin{tabular}[c]{@{}c@{}}RF \\ Chain \end{tabular}}
& \textbf{\begin{tabular}[c]{@{}c@{}}Alg. \\ Complex. \end{tabular}}
&  \textbf{\begin{tabular}[c]{@{}c@{}}Commun. \\ Delay \end{tabular}}
 &
\textbf{\begin{tabular}[c]{@{}c@{}}Objective \\ Function \end{tabular}}
& \textbf{AirComp}
& \textbf{\begin{tabular}[c]{@{}c@{}}FL \\ Task \end{tabular}}
\\ \hline
\multirow{4}{*}{\textbf{\begin{tabular}[c]{@{}c@{}}Digital \end{tabular}}}  & \cite{fed1}          &   \XSolidBrush       & +   & +                            & +++         & N/A      & \XSolidBrush                                                                  & Classification
\\ \cline{2-9}
& \cite{fed2,fed3}          &   \XSolidBrush        & +   & +                            & ++         & Loss Bound  & \XSolidBrush                                                                  & Classification
\\ \cline{2-9}
& \cite{fed4}   &  \XSolidBrush              & +       & +         & ++     & Loss Bound    & \XSolidBrush     &    Classification
\\ \cline{2-9}
                                                                                     & \cite{fed5}   &  Digital             & +++         & +++         & ++     & MSE    & \XSolidBrush     &    \XSolidBrush                                                                                                            \\ \hline
\multirow{4}{*}{\textbf{\begin{tabular}[c]{@{}c@{}}Analog \end{tabular}}}  &
\begin{tabular}[c]{@{}c@{}} \cite{air2,air3}\end{tabular} &            \XSolidBrush                                                                                              & +      & +                          & +                                                                   & Heuristic                                                                  & \Checkmark & Classification
\\ \cline{2-9}
& \cite{air4}       & Digital                                                                                                                            & +++                                                                & +++                                                                & +                                                                     & MSE                                                                    & \Checkmark & Classification
\\ \cline{2-9}
                                                                                     & \cite{air5}       &       Digital                                                    &           +++                                                       & ++                                                                  & +                                                             & Noise Variance                                                                                                                                     & \Checkmark  & Classification
\\ \cline{2-9}
                                                                                     & \textbf{Ours}                                                & Analog                                                                                                                          & +        & +                                                              & +                                                                      & Loss Bound                                                                   &    \Checkmark  & Object Detection                                      \\ \hline
\hline
\end{tabular}
}
\vspace{0.1in}
\label{Table.related_work}
\hspace{1cm}

The symbol ``+'' means low, ``++'' means moderate, ``+++'' means high.

The symbol ``\checkmark'' means functionality supported, ``\XSolidBrush" means functionality not supported.
\end{table*}

On the other hand, the key advantage of analog modulation \cite{air2,air3,air4,air5} over digital modulation arises from the ground-breaking idea of over-the-air computation (AirComp).
Specifically, if multiple users upload their local parameters simultaneously, a superimposed signal, which represents a weighted sum of individual model parameters, is observed at the edge server.
By performing minimum mean square error (MMSE) detection on the superimposed signal, an estimate of the global parameter vector can be obtained.
This significantly saves the transmission time since AirComp in fact exploits inter-user interference in the simultaneous user transmission \cite{ota1,ota2,ota3}, in oppose to interference suppression as in digital modulation.
Furthermore, the wireless channel can be viewed as a natural privacy-preserving mask for the released models, which injects random channel noises into the parameters \cite{ota3}.

However, due to channel fading and noise in wireless systems, AirComp employed in single-antenna systems \cite{air2,air3} could result in a large error in the estimation of global model parameters at the edge server, leading to slow convergence of FL.
As a remedy, adopting MIMO beamforming \cite{air4,air5} could reduce the parameter transmission error by aligning the beams carrying the local parameters' information to the same spatial direction.
Nonetheless, the current transmit and receive beamforming designs in MIMO AirComp systems involve exceedingly high radio frequency (RF) chain costs and high computational complexities \cite{air4,air5}, preventing their practical implementation.

In practice, both digital and analog modulation methods share the same goal, i.e., minimizing the training loss function.
However, due to the lack of an explicit form of the training loss function with respect to wireless designs, most works focus on other related objective functions such as mean square error (MSE) \cite{air4} and noise variance \cite{air5}.
Recently, the relationship between the training loss function and the wireless designs is derived in \cite{fed2,fed3,fed3b,iclr,air1}.
Nonetheless, the bounds in \cite{fed2,fed3,fed3b,iclr,air1} are only applicable if the global model parameters are perfectly broadcast to users.
For practical cases involving errors in the model broadcast phase, new training loss bounds are required to capture the training performance, which remains an open problem.

\subsection{Summary of Challenges and Contributions}

In summary, despite recent exciting development of edge FL techniques, e.g.,  \cite{fed1,fed2,fed3,fed3b,fed4,fed5,air2,air3,air4,air5}, a number of technical challenges remain to be overcome, including
\begin{itemize}
    \item[1)] \textbf{Analog beamforming design under federated learning settings}.
    Analog beamforming with a proper phase shift network design \cite{analog,analog2} can help reduce implementation costs compared with digital beamfroming in \cite{air4,air5}, which has not been studied in edge FL systems, yet.
    Since the analog beamformer is unit-modulus, it introduces a large number of intractable nonconvex and nonsmooth constraints.
    Furthermore, the analog beamformer in FL systems aims to minimize the FL training loss bound, which is a non-differentiable function of the beamforming design.
    This is different from the conventional analog beamforming designs in massive MIMO systems that aim to decode individual information from each user \cite{analog,analog2}.

    \item[2)] \textbf{Reduction of beamforming design complexities.}
    Most beamforming algorithms \cite{analog2, fed5,air4,air5} rely on the execution of the interior point method (IPM). Yet, since IPM involves the inversion of Hessian matrices, these algorithms are with high computational complexities requiring exceedingly long signal processing delay, especially when massive MIMO technique is applied.
    On the other hand, first-order methods \cite{acceleration1, acceleration3, yurii} (i.e., no inversion of Hessian matrices) could lead to significantly lower computational complexities than that of the IPM.
    However, they cannot directly handle the analog beamforming design problem due to the non-differentiable training loss bound and unit-modulus constraints.

    \item[3)] \textbf{Verification of robustness in more complex learning tasks}.
    Existing algorithms in \cite{fed1,fed2,fed3,fed3b,fed4,fed5,air2,air3,air4,air5} are mainly tested on simple image classification tasks (e.g., recognition of handwritten digits).
    Experiments on more complex and closer-to-reality tasks, such as 3D object detection \cite{SECOND,shaoshuai} in vehicle-to-everything (V2X) autonomous driving systems, are needed to verify the robustness of edge FL.
    However, the associated implementation requires a software cluster including scenario generation, multi-modal data generation, multi-sensor calibration, multi-vehicle synchronization and coordinate transformation, label generation, and object detection.
\end{itemize}

To fill the research gap, this paper proposes the unit-modulus AirComp (UMAirComp) framework for edge FL in MIMO communication systems, as shown in Fig.~1.
The UMAirComp framework consists of multiple edge users with local sensing data (e.g., autonomous driving cars generate camera images and light detection and ranging (LiDAR) point clouds of the environment), an edge server for performing model aggregation, and communication interfaces for exchanging model parameters.
Specifically, the edge users update their local models using the local training datasets separately.
Then, the trained model parameters are transformed into signals via analog modulation and uploaded to the server.
To reduce the implementation cost of RF chains, the edge server does not process the received signals at the baseband.
Instead, it applies a phase shift network (either a fully-connected structure or a partially-connected structure) in the RF domain for global model aggregation and broadcasting.
All users decode model parameters from the received broadcast signals via analog demodulation.
The advantages of UMAirComp and contributions of this paper are summarized below:

\begin{itemize}
    \item[1)] The UMAirComp at the server significantly reduces the required implementation cost of RF chains in MIMO FL systems, thereby reducing the hardware and energy costs.
    To understand how UMAirComp works, the training loss of UMAirComp framework is proved to be upper bounded by a monotonic increasing function of the maximum MSE of the model parameters' estimation.

    \item[2)] Despite the UMAirComp problem being highly nonconvex, two large-scale optimization algorithms, termed penalty alternating minimization (PAM) and accelerated gradient projection (AGP), are developed for fully-connected UMAirComp and partially-connected UMAirComp, respectively.
    The learning performance of the proposed PAM and AGP is shown to outperform other benchmark schemes.
    In particular, the AGP algorithm is $100$x faster than that of the state-of-the-art optimization algorithms.

    \item[3)] We implement the UMAirComp edge FL scheme for 3D object detection with multi-vehicle point-cloud datasets in CARLA simulation platform \cite{carla}. To the best of our knowledge, this is the first attempt that edge FL is demonstrated in a V2X auto-driving simulator with a close-to-reality task.
\end{itemize}

\begin{figure}[!t]
 \centering
\includegraphics[width=0.8\textwidth]{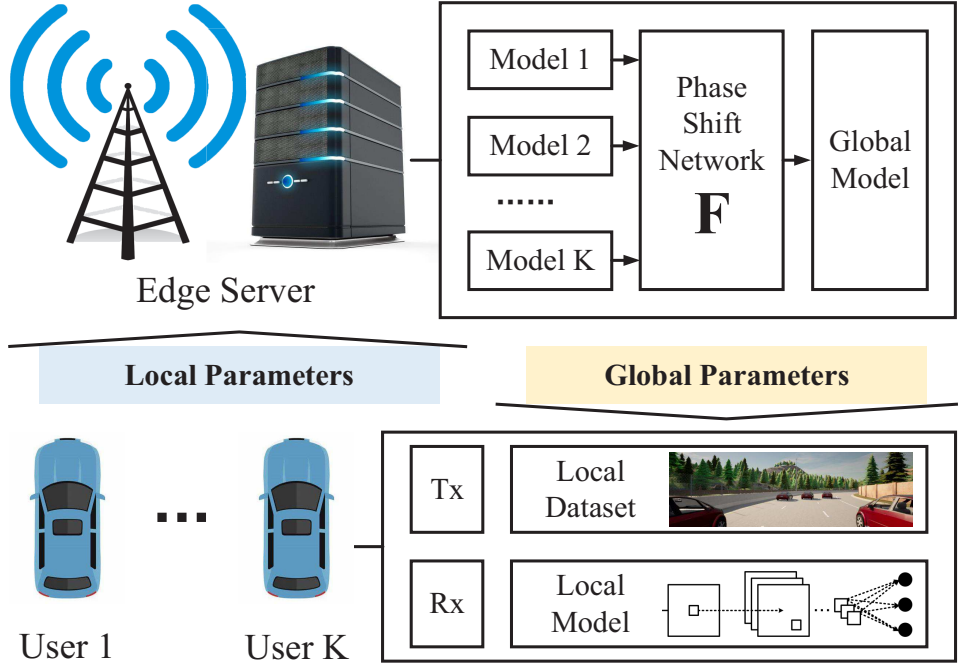}
  \caption{Illustration of the UMAirComp edge FL system.}
\end{figure}

\subsection{Notation}

Italic letters, lowercase and uppercase bold letters represent scalars, vectors, and matrices, respectively.
Curlicue letters stand for sets and $|\cdot|$ is the cardinality of a set.
The operators $\|\cdot\|_2,(\cdot)^{T},(\cdot)^{H},(\cdot)^{-1},\lambda_{\mathrm{max}}(\cdot),\lambda_{\mathrm{min}}(\cdot),\mathrm{Null}(\cdot),\mathrm{Rank}(\cdot),\mathrm{vec}(\cdot)$ are the $\ell_2$-norm, transpose, Hermitian, inverse, largest eigenvalue, smallest eigenvalue, null space, rank, and vectorization of a matrix, respectively.
The operators $\partial f$ and $\nabla f$ are the partial derivative and the gradient of the function $f$.
The function $[x]^+=\mathrm{max}(x,0)$, $\mathrm{Re}(x)$ takes the real part of $x$, $\mathrm{Im}(x)$ takes the imaginary part of $x$, $\mathrm{conj}(x)$ takes the conjugate of $x$, and $|x|$ is the modulus of $x$.
$\mathbf{I}_{N}$ denotes the $N\times N$ identity matrix,
$\mathbf{1}_{N}$ represents the $N\times 1$ all-ones vector, and $\mathbf{A}\succeq \mathbf{B}$ represents $\mathbf{A}-\mathbf{B}$ being positive semidefinite.
Finally, $\mathrm{j}=\sqrt{-1}$, $\mathbb{E}(\cdot)$ denotes the expectation of a random variable and $\mathcal{O}(\cdot)$ is the big-O notation standing for the order of arithmetic operations.

\section{Edge Federated Learning with AirComp}

\setcounter{secnumdepth}{4}
We consider an edge FL system shown in Fig.~1, which consists of an edge server equipped with $N$ antennas and $K$ single-antenna mobile users. The dataset and model parameter vector at user $k$ are denoted as $\mathcal{D}_k$ and $\mathbf{x}_k\in\mathbb{R}^{M\times 1}$, respectively.
Mathematically, the FL procedure aims to solve the following optimization problem:
\begin{align}\label{FL}
\mathop{\mathrm{min}}_{\substack{\{\mathbf{x}_k\},\bm{\theta}}}
\quad&\underbrace{\frac{1}{\sum_{k=1}^{K}|\mathcal{D}_k|}\sum_{k=1}^K\sum_{\mathbf{d}_{k,l}\in\mathcal{D}_k}\Theta(\mathbf{d}_{k,l}, \bm{\theta})}_{:=\Lambda(\bm{\theta})}\quad
\nonumber\\
\mathrm{s.t.}\quad&\mathbf{x}_1=\cdots=\mathbf{x}_K= \bm{\theta},
\end{align}
where $\Theta(\mathbf{d}_{k,l}, \bm{\theta})$ is the loss function corresponding to a single sample $\mathbf{d}_{k,l}$ ($1\leq l \leq |\mathcal{D}_k|$) in $\mathcal{D}_k$ given parameter vector $\bm{\theta}$, while $\Lambda(\bm{\theta})$ denotes the global loss function to be minimized.

The objective function can be rewritten into a separable form $\Lambda(\bm{\theta})=\sum_{k=1}^K\alpha_k\lambda_k(\mathbf{x}_k)$, where $\lambda_k(\mathbf{x}_k)=\frac{1}{|\mathcal{D}_k|}\sum_{\mathbf{d}_{k,l}\in\mathcal{D}_k}\Theta(\mathbf{d}_{k,l}, \mathbf{x}_k)$ is the loss function at the $k$-th user and $\alpha_k=\frac{|\mathcal{D}_k|}{\sum_{l=1}^{K}|\mathcal{D}_l|}$.
Therefore, the training of FL model parameters (i.e., solving \eqref{FL}) in the considered edge system is naturally a distributed and iterative procedure, where each iteration involves four steps:
1) updating the local parameter vectors $(\mathbf{x}_1,\cdots,\mathbf{x}_K)$ for minimizing $(\lambda_1,\cdots,\lambda_K)$ with respect to $\{\mathcal{D}_1,\cdots,\mathcal{D}_K\}$ at users $(1,\cdots,K)$, respectively;
2) transforming the local parameters $(\mathbf{x}_1,\cdots,\mathbf{x}_K)$ into symbols $(\mathbf{s}_1,\cdots,\mathbf{s}_K)$ for uploading;
3) aggregating $(\mathbf{s}_1,\cdots,\mathbf{s}_K)$ in an analog manner at the edge server and broadcasting the results to the users.
4) receiving the signals $(\mathbf{y}_1,\cdots,\mathbf{y}_K)$ at the users and transforming $(\mathbf{y}_1,\cdots,\mathbf{y}_K)$ into parameters $(\mathbf{x}_1,\cdots,\mathbf{x}_K)$ for next-round updates.
The above four steps are elaborated below.

\begin{itemize}
\item[1)] Step 1: Let $\mathbf{x}^{[i]}_k(0)\in\mathbb{R}^{M\times 1}$ be the local parameter vector at user $k$ at the beginning of the $i$-th iteration ($i\geq0$ and $\mathbf{x}^{[0]}_k(0)=\bm{\theta}^{[0]}$).
To update $\mathbf{x}^{[i]}_k(0)$, user $k$ minimizes the loss function $\lambda_k(\mathbf{x}_k)$ via gradient descent\footnote{If $|\mathcal{D}_k|$ is large, stochastic gradient descent can be adopted to accelerate the training speed.} as
\begin{align}\label{sgd}
&\mathbf{x}^{[i]}_k(\tau+1)=\mathbf{x}^{[i]}_k(\tau)-\frac{\varepsilon}{|\mathcal{D}_k|}\sum_{\mathbf{d}_{k,l}\in\mathcal{D}_k}\nabla\Theta\left(\mathbf{d}_{k,l},\mathbf{x}^{[i]}_k(\tau)\right),
\end{align}
where $\varepsilon$ is the step-size and $\tau$ is from $0$ to $E-1$ with $E$ being the number of local updates.

\item[2)] Step 2: All users upload $\{\mathbf{x}^{[i]}_k(E)|\forall k\}$ to the edge server.
Specifically, user $k$ modulates its local parameter vector $\mathbf{x}^{[i]}_k(E)$ into a complex vector $\mathbf{s}^{[i]}_k\in\mathbb{C}^{S\times 1}$ with $S=M/2$ in an analog manner as in \cite{air2,air3}, where
\begin{align}
\mathbf{s}^{[i]}_k=&
\frac{\sqrt{p_k^{[i]}}\,\mathrm{exp}(\mathrm{j}\varphi_k^{[i]})}{\sqrt{2\eta^{[i]}}}\,\Big[x_{k,1}^{[i]}(E)+\mathrm{j}\,x_{k,2}^{[i]}(E),\cdots,x_{k,M-1}^{[i]}(E)+\mathrm{j}\,x_{k,M}^{[i]}(E)\Big]^T. \label{sk}
\end{align}
In \eqref{sk}, $x_{k,m}^{[i]}(E)$ is the $m$-th element of $\mathbf{x}_k^{[i]}(E)$.
The scaling factor $\eta^{[i]}=\frac{1}{K}\sum_{k=1}^K\frac{1}{M}\|\mathbf{x}_k^{[i]}(E)\|_2^2$ such that the average power of $\mathbf{s}_k^{[i]}$ is $\frac{1}{S}\mathbb{E}\left[\|\mathbf{s}_k^{[i]}\|_2^2\right]=p_k^{[i]}$.\footnote{
In practice, the users would adopt $\eta^{[i-1]}$ sent from the server in the last iteration, which is a good approximation for the actual $\eta^{[i]}$ in the current iteration.}
The transmit power is $p_k^{[i]}\leq P_0$ with $P_0$ being the maximum transmit power at each user and the phase shift is $\varphi^{[i]}_k\in[0,2\pi]$.
To facilitate the subsequent derivations, we define the transmitter design $t_k^{[i]}=\sqrt{p_k^{[i]}}\,\mathrm{exp}(\mathrm{j}\varphi_k^{[i]})$ for all $(i,k)$ and $\{p_k^{[i]},\varphi_k^{[i]}\}$ can be recovered from $\{t_k^{[i]}\}$.

\item[3)] Step 3: The received signal $\mathbf{R}^{[i]}\in\mathbb{C}^{N\times S}$ at the server is
\begin{align}\label{Ri}
&\mathbf{R}^{[i]}=\sum_{k=1}^K\mathbf{h}_k^{[i]}(\mathbf{s}_{k}^{[i]})^T+\mathbf{Z}^{[i]},
\end{align}
where $\mathbf{h}_k^{[i]}\in\mathbb{C}^{N\times 1}$ is the uplink channel vector from user $k$ to the server and $\mathbf{Z}^{[i]}\in\mathbb{C}^{N\times S}$ is the matrix of the additive white Gaussian noise with covariance matrix
$\mathbb{E}\left[\mathrm{vec}(\mathbf{Z}^{[i]})\mathrm{vec}(\mathbf{Z}^{[i]})^H\right]=\sigma^2_b\mathbf{I}_{NS}$ ($\sigma^2_b$ is the noise power at the server).
Upon receiving the superimposed signal, the server processes $\mathbf{R}^{[i]}$ using a function $\Psi^{[i]}(\cdot)$ and broadcasts $\Psi^{[i]}(\mathbf{R}^{[i]})\in\mathbb{C}^{N\times S}$ to all the users.

\item[4)] Step 4: The received signal at user $k$ is
\begin{align}
&\left(\mathbf{y}^{[i]}_k\right)^T=\left(\mathbf{g}_k^{[i]}\right)^H\Psi^{[i]}\left(\mathbf{R}^{[i]}\right)+\left(\mathbf{n}^{[i]}_k\right)^T, \label{yk}
\end{align}
where $\mathbf{g}_k^{[i]}\in\mathbb{C}^{N\times 1}$ is the downlink channel vector\footnote{It is assumed that $\{\mathbf{h}_k^{[i]},\mathbf{g}_k^{[i]}\}$ are quasi-static flat-fading channels during the model exchange procedure of each FL iteration.} from the server to user $k$ and $\mathbf{n}_k^{[i]}\in\mathbb{C}^{S\times 1}$ is the vector of the additive white Gaussian noise with covariance matrix $\sigma^2_k\mathbf{I}_S$ ($\sigma^2_k$ is the noise power at user $k$).
User $k$ demodulates $\mathbf{y}_k^{[i]}$ and sets the local parameter vector for the $(i+1)$-th iteration as
\begin{align}
\mathbf{x}^{[i+1]}_k(0)
=&
\sqrt{2\eta^{[i]}}\,\Big[\mathrm{Re}(r_k^{[i]}y^{[i]}_{k,1}),\mathrm{Im}(r^{[i]}_ky^{[i]}_{k,1}),\cdots,
\mathrm{Re}(r^{[i]}_ky^{[i]}_{k,S}),\mathrm{Im}(r^{[i]}_ky^{[i]}_{k,S})\Big]^T, \label{nextround}
\end{align}
where $r^{[i]}_k\in\mathbb{C}$ is a receive coefficient applied to $\mathbf{y}^{[i]}_k$ and $y^{[i]}_{k,s}$ is the $s$-th element of $\mathbf{y}^{[i]}_k$.
This completes one federated learning round and we set $i\leftarrow i+1$.
\end{itemize}

The entire procedure stops when $i=R$ with $R$ being the number of federated learning rounds.

\section{Proposed UMAirComp Edge FL Framework}

In existing AirComp schemes, e.g., \cite{air2,air3,air4,air5}, $\Psi^{[i]}$ is implemented using baseband signal processing techniques.
Thus, the superimposed signals of all receive antennas at the edge server are first combined via a vector in the digital baseband and then broadcast to users via another vector in the baseband.
Therefore, the associated implementation cost could be high if massive number of antennas are deployed.
In the following, the UMAirComp scheme, where the functions $\{\Psi^{[i]}\}$ are implemented in analog domain, is proposed.
This helps to reduce both the required implementation costs and power consumptions while achieving excellent system performance.

\subsection{UMAirComp Framework}

The proposed UMAirComp scheme employs multiple phase shifters to generate the updated global model parameters from the local ones.
Specifically, we propose to aggregate and forward the signal via
\begin{align}
&\Psi^{[i]}(\mathbf{R}^{[i]})=\sqrt{\gamma}\,\mathbf{F}^{[i]}\mathbf{R}^{[i]}, \label{Psi}
\end{align}
where $\gamma>0$ is the power scaling factor at the edge server and $\mathbf{F}^{[i]}\in\mathbb{C}^{N\times N}$ is the phase shift matrix.
The analog domain implementation of $\mathbf{F}^{[i]}$ is shown in Fig.~2, which inserts an analog beamformer in the RF domain immediately after the low-noise amplifiers and bandpass filter.
This is in contrast to the digital domain implementation, which inserts a complete RF receiver chain and a high-rate high-resolution analog-to-digital-converter (ADC) unit for each antenna.
The analog beamformer constructs linear combinations of slightly delayed antenna signals and results in $N_{\rm{ADC}}$ output signals.
Depending on the value of $N_{\rm{ADC}}$, the implementation can be categorized into fully-connected (Fig.~2a) or partially-connected (Fig.~2b) \cite{analog,analog2}.

For the fully-connected structure shown in Fig.~2a, we have $N_{\rm{ADC}}>1$.
Each entry $F^{[i]}_{l,l'}$ of $\mathbf{F}^{[i]}$ corresponds to the phase delay introduced by the phase shifter for the $l$-th receive antenna and the $l'$-th output signal.
Note that amplitude changes are not possible.
This results in unit-modulus constraints on all elements of the analog beamforming matrices $\{\mathbf{F}^{[i]}\}$, i.e.,
\begin{align}
&|F_{l,l'}^{[i]}|=1,\quad\forall l,l'. \label{APpower}
\end{align}
The output signals are stacked into vectors, delayed for a guard period (e.g., using RF sampling or RF delay line systems), converted to downlink signals, amplified by $\gamma$, and transmitted to users.

For the partially-connected structure shown in Fig.~2b, we have $N_{\rm{ADC}}=1$.
Each receive antenna element is processed by a phase shifter and combined into a single output signal.
The output signal is downconverted to a baseband signal, which is sampled and quantized into bits via an ADC unit.
After delaying for a guard period, the signal is upconverted to a passband signal and filtered via another phase shift vector before transmission.
Hence, apart from the unit-modulus constraint \eqref{APpower}, $\mathbf{F}^{[i]}$ should also satisfy $\mathrm{Rank}\left(\mathbf{F}^{[i]}\right)=1$.
Note that the required number of total phase shifters is only $2N$ for the partially-connected structure.

\begin{figure*}[!t]
 \centering
\subfigure[]{ \includegraphics[height=0.22\textwidth]{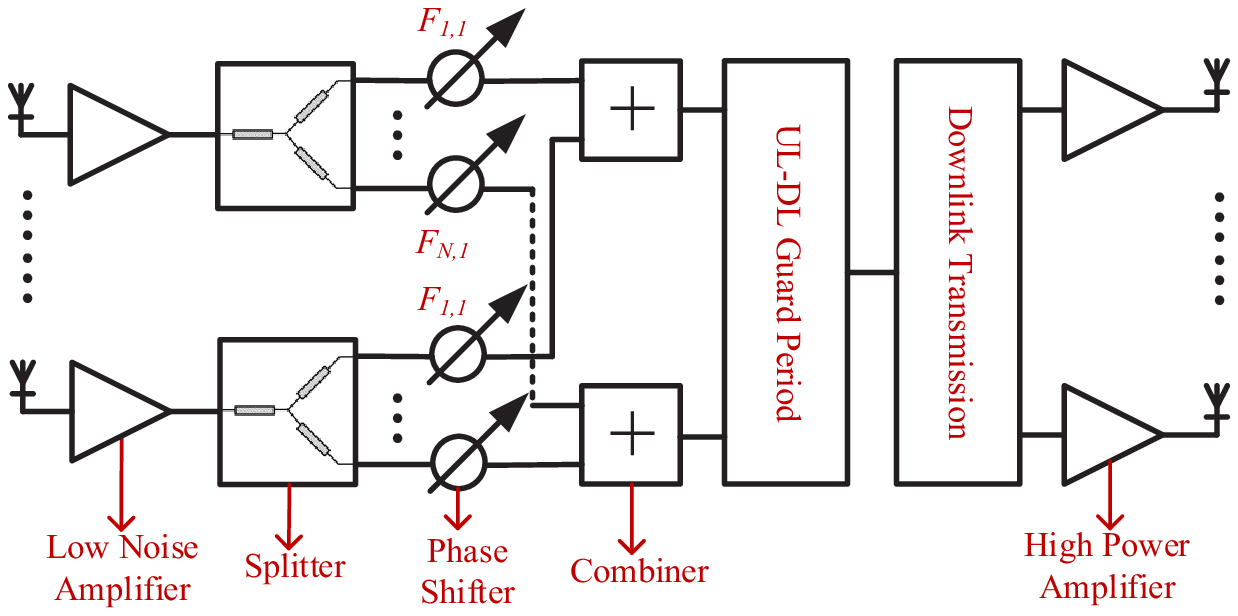}}
\subfigure[]{ \includegraphics[height=0.205\textwidth]{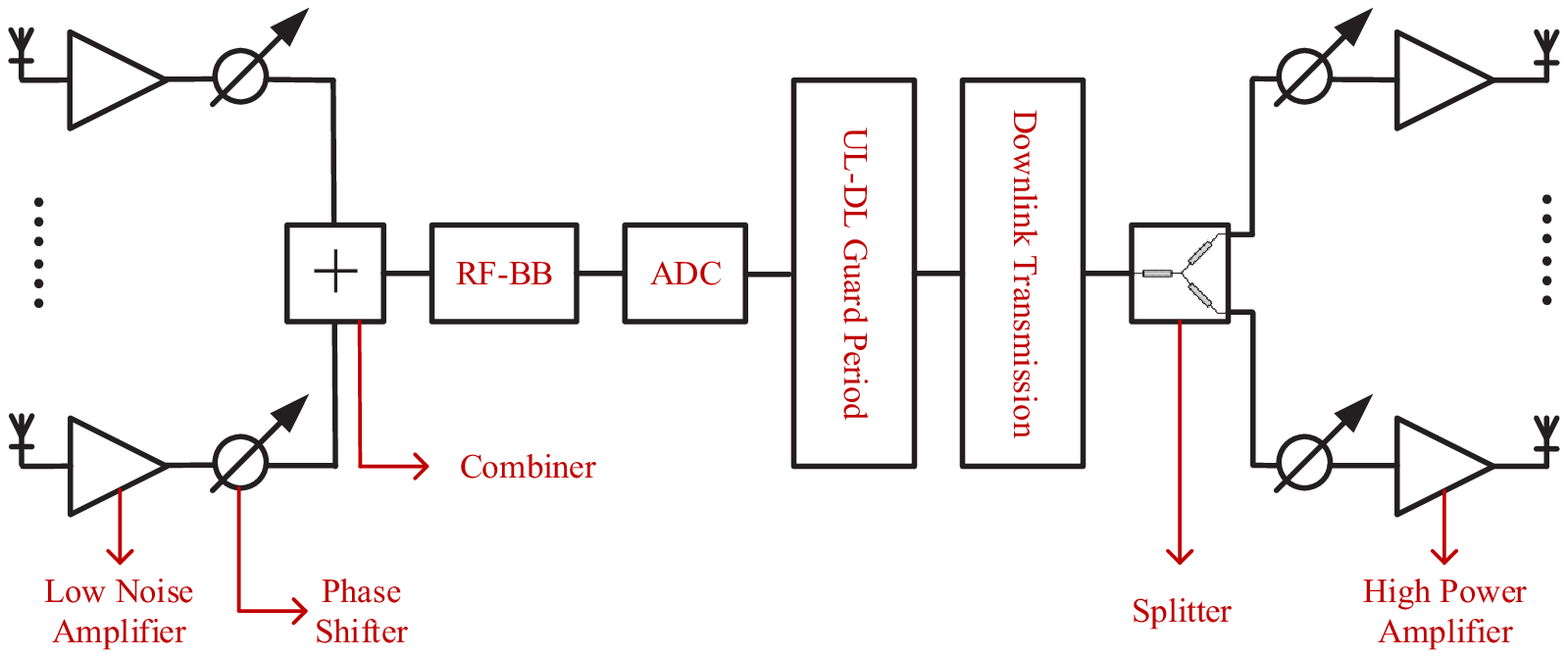}}
  \caption{Illustration of the phase shift network at the server: a) fully-connected structure; b) partially-connected structure.}
\end{figure*}

Based on equations \eqref{sgd}--\eqref{Psi}, the MSE between the received local parameter $\mathbf{x}^{[i+1]}_k$ and the target local parameter $\bm{\theta}^{[i]}$ at the $(i+1)$-th FL iteration for the proposed UMAirComp scheme is
\begin{align}\label{mse}
&\mathbb{MSE}_k^{[i]}\left(\mathbf{F}^{[i]}, \{t_k^{[i]}, r_k^{[i]}\}\right)
=\mathbb{E}\left[\Big\|\mathbf{x}^{[i+1]}_k(0)-\bm{\theta}^{[i]}\Big\|_2^2\right]
\nonumber\\
&=2S\eta^{[i]}\Bigg[\gamma\sum_{j=1}^K\Big|r_k^{[i]}(\mathbf{g}_k^{[i]})^H\mathbf{F}^{[i]}\mathbf{h}_j^{[i]}t_j^{[i]}-\alpha_j\Big|^2
+\gamma\sigma^2_b\|r_k^{[i]}(\mathbf{g}_k^{[i]})^H\mathbf{F}\|_2^2
+\sigma_k^2|r_k^{[i]}|^2\Bigg],
\end{align}
where the target model parameter
\begin{align}
&\bm{\theta}^{[i]}=\sum_{k=1}^K\alpha_k\mathbf{x}_k^{[i]}(E) \label{global}
\end{align}
is equivalent to the gradient descent of the objective function of \eqref{FL}.
Note that the $\mathbb{MSE}_k^{[i]}$ is the total MSE between the obtained and target models, which consists of $M=2S$ parameters.
In the following, we will also use ${\mathbb{MSE}^{[i]}_k/(2S\eta^{[i]})}$ as a performance metric, which is the normalized MSE representing the relative distortion with respect to each parameter.

\subsection{Training Loss Analysis}

Due to the model distortion $\mathbb{MSE}_k^{[i]}$ in \eqref{mse}, the training loss $\Lambda(\mathbf{x}_k^{[R]}(0))$ in \eqref{FL} after the UMAirComp edge FL terminates is random.
Its expectation $\mathbb{E}[\Lambda(\mathbf{x}_k^{[R]}(0))]$, where the expectation is taken over channel noises and model parameters, would be greater than $\Lambda(\bm{\theta}^{*})$ in wired FL systems, where $\bm{\theta}^*$ denotes the optimal solution of $\bm{\theta}$ to \eqref{FL}.
Hence, we can use their difference, namely $\mathbb{E}[\Lambda(\mathbf{x}_k^{[R]}(0))]-\Lambda(\bm{\theta}^*)$, as a metric to capture the degradation on training loss.
Before establishing the main results, we first introduce the assumption imposed on the loss function.

\begin{assumption}
(i) The function $\Lambda(\mathbf{x})$ is $\mu$-strongly convex.
(ii) The function $\frac{1}{|\mathcal{D}_k|}\sum_{\mathbf{d}_{k,l}\in\mathcal{D}_k}\Theta(\mathbf{d}_{k,l}, \mathbf{x})$ is twice differentiable and satisfies $\frac{1}{|\mathcal{D}_k|}\sum_{\mathbf{d}_{k,l}\in\mathcal{D}_k}\nabla^2_{\mathbf{x}}\Theta(\mathbf{d}_{k,l}, \mathbf{x})\preceq L\mathbf{I}$.
\end{assumption}

Under Assumption 1, the relationship between $\Lambda(\mathbf{x}_k^{[R]}(0))$ and $\Lambda(\bm{\theta}^{*})$ is summarized in the following theorem.
\begin{theorem}
With $E=1$ and $\varepsilon=\frac{1}{L}$, the UMAirComp scheme satisfies
\begin{align}
\mathbb{E}\left[\Lambda(\mathbf{x}_k^{[R]}(0))\right]-\Lambda(\bm{\theta}^*)
\leq
\sum_{i=0}^{R-1}A^{[i]}\,\mathop{\mathrm{max}}_{k=1,\cdots,K}\mathbb{MSE}^{[i]}_k,
\end{align}
as $R\rightarrow+\infty$, where
\begin{align}
A^{[i]}=\frac{L}{2}\left(3+2\sum_{j=1}^K\alpha_j^2\right)
\left(1-\frac{\mu}{L}\right)^{R-1-i}.
\end{align}
\end{theorem}
\begin{proof}
First, it can be proved that $\mathbb{E}\left[\|\Delta \mathbf{x}^{[i]}\|^2_2\right]\leq
\left(3+2\sum_{j=1}^K\alpha_j^2\right)
\,\mathop{\mathrm{max}}_{k}\mathbb{MSE}^{[i]}_k$, where $\Delta \mathbf{x}^{[i]}=\mathbf{x}_k^{[i+1]}(0)-\left[\mathbf{x}_k^{[i]}(0)-\varepsilon\nabla_{\mathbf{x}}\Lambda(\mathbf{x}_k^{[i]}(0))\right]$.
Second, using Assumption 1 and Lipschitz conditions, the relationship between $\Lambda(\mathbf{x}_k^{[i+1]}(0))-\Lambda(\mathbf{x}_k^{[i]}(0))$ and $\mathbb{E}\left[\|\Delta \mathbf{x}^{[i]}\|^2_2\right]$ is obtained.
Lastly, applying the former relationship recursively, the relationship between $\mathbb{E}[\Lambda(\mathbf{x}_k^{[i+1]}(0))]-\Lambda(\bm{\theta}^*)$ and $\mathbb{E}\left[\|\Delta \mathbf{x}^{[i]}\|^2_2\right]$ is obtained.
More details can be found in Appendix A.
\end{proof}

Theorem 1 shows a diminishing $A^{[i]}\rightarrow 0$ for a large $R-1-i$, meaning that the impact from earlier FL iterations vanishes as the edge FL continues.
On the other hand, if $\mathbb{MSE}^{[i]}_k\rightarrow 0$ for all $k$, then $\Lambda(\mathbf{x}_k^{[R]})$ is an unbiased estimate of $\Lambda(\bm{\theta}^*)$.
This demonstrates the effectiveness of UMAirComp in the asymptotic region.
The convexity and smoothness in Assumption 1 have been adopted in most loss bound analysis of FL (e.g., \cite{mse1,iclr,air1,yurii}).
Although it seems to be restrictive for realistic applications, analysis under Assumption 1 could provide important insights of the behavior of UMAirComp in nonconvex cases.

The analysis can be generalized to the case with multiple local epochs (i.e., $E>1$).
Specifically, in addition to Assumption 1, we further impose the following assumption.

\begin{assumption}
$\|\frac{1}{|\mathcal{D}_k|}\sum_{\mathbf{d}_{k,l}\in\mathcal{D}_k}\nabla_{\mathbf{x}}\Theta(\mathbf{d}_{k,l}, \mathbf{x})\|_2^2\leq G^2$ for some $G>0$ and all $k$.
\end{assumption}

In practice, the above assumption can always be satisfied with a sufficiently large $G$, which represents the maximum gradient norm in the back propagation procedure.
Otherwise, the exploding gradients will lead to the unstable training of deep networks.
Let $\Lambda^*$ denotes the minimum training loss in \eqref{FL} over all users' datasets and $\lambda_k^*$ denotes the minimum training loss over the $k$-th user's dataset.
Then, the degree of heterogeneity of the edge FL system is defined as $\Gamma=\Lambda^*-\sum_{k=1}^K\alpha_k\lambda_k^*$.
It can be seen that $\Gamma=0$ if all users have the same dataset.
Then the following convergence result can be established.

\begin{theorem}
With $\varepsilon=\frac{2}{\mu(\nu+iE+\tau)}$ and $\nu=\mathrm{max}(\frac{8L}{\mu},E)$, the UMAirComp scheme satisfies
\begin{align}
\mathbb{E}\left[\Lambda(\mathbf{x}_k^{[R]}(0))\right]-\Lambda(\bm{\theta}^*)
\leq
\frac{2L\,\mathrm{max}(4C,\mu^2\nu\|\bm{\theta}^{[0]}-\bm{\theta}^* \|_2^2)}{\mu^2(RE+\nu)},
\end{align}
where
\begin{align}
&C=8E^2G^2+6L\Gamma+\frac{\mu^2(\nu+RE)^2}{4}\,\mathop{\mathrm{max}}_{\forall i,k}\,\mathbb{MSE}_k^{[i]}.
\end{align}
\end{theorem}
\begin{proof}
First, we define a virtual sequence and a recursive bound can be obtained for this sequence based on \cite[Lemma 1]{iclr}.
Second, the recursive bound is further simplified by bounding the model distortion and gradient deviation.
Lastly, a non-recursive bound is obtained using the induction method in \cite[Appendix B]{air1} and the results follow immediately.
More details can be found in Appendix B.
\end{proof}

It can be seen from Theorem 2 that if we increase the number of local epochs $E$ or the degree of heterogeneity $\Gamma$, the term $C$ would increase and the gap between $\mathbb{E}\left[\Lambda(\mathbf{x}_k^{[R]}(0))\right]$ and $\Lambda(\bm{\theta}^*)$ becomes larger.
This implies that more local training epochs or more diverse datasets would make it more difficult for the proposed UMAirComp federated learning to converge.

\subsection{Problem Formulation}

Ideally, the optimization of $\{\mathbf{F}^{[i]},r_k^{[i]},t_k^{[i]}\}$ should be performed to minimize the expectation of training error, i.e., $\mathbb{E}[\Lambda(\mathbf{x}_k^{[R]}(0))]$.
However, its analytical expression is usually challenging to derive.
As a compromise, we resort to the minimization of the upper bound obtained from Theorem 1 or Theorem 2, which at least guarantees the worst-case training loss performance.
Minimizing the training loss bound is equivalent to minimizing $\sum_{i=0}^{R}A^{[i]}\,\mathrm{max}_k\mathbb{MSE}^{[i]}_k$ for the case of $E=1$ and $\mathop{\mathrm{max}}_{\forall i,k}\mathbb{MSE}_k^{[i]}$ for the case of $E>1$.
But no matter $E=1$ or $E>1$, the objective function can be decoupled for each iteration and the minimization at the $i$-th FL iteration, $\forall i$, is given by
\begin{subequations}
\begin{align}
\mathcal{P}:\mathop{\mathrm{min}}_{\substack{\mathbf{F},\,\{r_k,t_k\}}}
\quad&
\mathop{\mathrm{max}}_{k=1,\cdots,K}~\Bigg[\gamma\sum_{j=1}^K\Big|r_k\mathbf{g}_k^H\mathbf{F}\mathbf{h}_jt_j-\alpha_j\Big|^2
+\gamma\sigma^2_b\|r_k(\mathbf{g}_k)^H\mathbf{F}\|_2^2+\sigma_k^2|r_k|^2\Bigg]
\label{P0}
\\
\quad\quad\quad\mathrm{s. t.}\quad\quad
&\mathbf{F}\in\mathcal{F}, \label{unit}
\\
&|t_k|^2\leq P_0,\quad k=1,\cdots,K, \label{iou}
\end{align}
\end{subequations}
where the FL iteration index $i$ is omitted.
The constraint  \eqref{unit} is the beamforming constraints at the server, with the feasible set
\begin{align}
&\mathcal{F}=\left\{
\begin{aligned}
&\{\mathbf{F}: |F_{l,l'}|=1,\, \forall l,l'\}, \quad &\textrm{FC}&
\\
&\{\mathbf{F}: \mathrm{Rank}\left(\mathbf{F}\right)=1, \quad |F_{l,l'}|=1,\, \forall l,l'\}, \quad &\textrm{PC}&
\end{aligned}
\right.
,
\label{setF}
\end{align}
where FC represents the fully connected case and PC represents the partially connected case.

The constraint \eqref{iou} is the power constraints at users obtained from Section II.
It can be seen that the key to minimizing the training loss upper bound of UMAirComp is to minimize the \emph{maximum MSE} instead of the average MSE.
Problem $\mathcal{P}$ is NP-hard due to the unit-modulus constraints \cite{analog,analog2}.
In addition, the coupling between variables $\{r_k\}$, $\{t_k\}$, and $\mathbf{F}$ makes the problem nonlinear and nonconvex.
Furthermore, the large dimensions of $\mathbf{F}$ and $\{r_k,t_k\}$ call for low-complexity algorithms in the scenario with massive antennas and users.

\subsection{Practical Implementation}

\begin{figure*}[!t]
 \centering
 \includegraphics[width=0.95\textwidth]{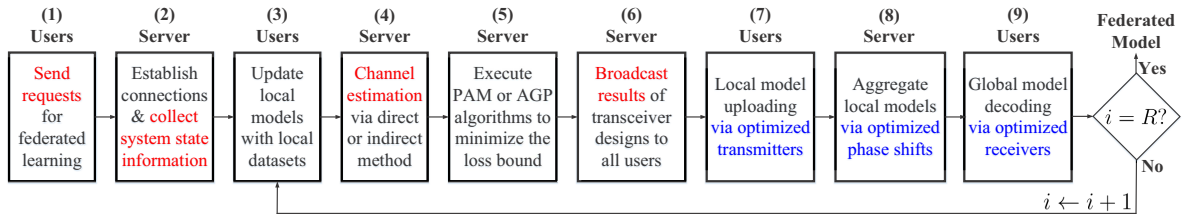}
 \caption{Flowchart of the proposed UMAirComp federated learning including system state information estimation, local training, channel state information estimation, results broadcast, model uploading and aggregation. The text in red represents the information exchange procedure and the text in blue represents the signal processing procedure.}
\end{figure*}

In practice, parameters $\gamma,\{\alpha_k\},\sigma_b^2,\{\sigma_k^2\},P_0,\{\mathbf{h}_k\},\{\mathbf{g}_k\}$ should be collected at the edge server before solving $\mathcal{P}$.
These information can be categorized into two types: system state information and channel state information.
The system state information includes the power scaling factor $\gamma$ at the server, parameter aggregation weights $\{\alpha_k\}$ at the server, noise powers $\{\sigma_b^2,\sigma_k^2\}$ at the server and users, and maximum transmit power $P_0$ at users.
This information can be viewed as constants during the entire federated learning procedure.
They are collected only once before the first federated learning iteration and pre-stored at the edge server for subsequent use.
On the other hand, the channel state information, including uplink channels $\{\mathbf{h}_k\}$ and downlink channels $\{\mathbf{g}_k\}$, needs to be estimated using pilot signals before each federated learning round.
Since the digital baseband processor does not have access to the individual antenna signals, the channels in UMAirComp systems should be estimated using low-rate low-resolution ADCs or the least squares fitting \cite[Section VI]{analog2}.

Based on the above discussions, the entire signalling procedure for obtaining the necessary information to solve the UMAirComp FL problem is shown in Fig.~3.
It can be seen that there are four messages (highlighted in red) that have to be exchanged between the server and users:
\begin{itemize}
\item[1)] Users to server: requests for the federated learning;
\item[2)] Users to server: system state information including the number of antennas and users, transmit power budgets, sizes of datasets, noise powers;
\item[3)] Users to server: pilot signals for channel estimation;
\item[4)] Server to users: solutions of transceivers generated by the optimization algorithms.
\end{itemize}

\section{Penalty Alternating Minimization for Fully-Connected UMAirComp}

In this section, the UMAirComp with fully-connected structure is studied, where the feasible set $\mathcal{F}$ is given in the first line of \eqref{setF}.
A PAM algorithm with two layers of iterations, i.e., an outer-layer iteration and an inner-layer iteration, will be proposed to optimize the system performance.
Below we first introduce the outer-layer iteration.

\subsection{Outer-Layer Iteration}

To resolve the coupling between variables $\{r_k\}$, $\{t_k\}$, and $\mathbf{F}$, this paper adopts an alternating optimization framework \cite{beck}, which optimizes one design variable at a time with others being fixed.
Starting with an initial solution $\{\mathbf{F}^{(0)},r_k^{(0)},t_k^{(0)}\}$, the entire procedure for solving the problem $\mathcal{P}$ for the $(n+1)$-th outer iteration, $\forall n$, can be elaborated below:
\begin{subequations}
\begin{align}
\mathbf{F}^{(n+1)}=\mathop{\mathrm{arg~min}}_{\substack{\mathbf{F}}}
&
~\mathop{\mathrm{max}}_{\forall k}~\Bigg(\sum_{j=1}^K\Big|r_k^{(n)}\mathbf{g}_k^H\mathbf{F}\mathbf{h}_jt_j^{(n)}-\alpha_j\Big|^2
+\sigma^2_b\|r_k^{(n)}\mathbf{g}_k^H\mathbf{F}\|_2^2\Bigg) \nonumber\\
\quad\quad\mathrm{s. t.}\quad
&|F_{l,l'}|=1,\quad \forall,l,l', \label{problemF}
\\
\{r_k^{(n+1)}\}=\mathop{\mathrm{arg~min}}_{\substack{\{r_k\}}}
&
~\mathop{\mathrm{max}}_{\forall k}
\Bigg(\gamma\sum_{j=1}^K\Big|r_k\mathbf{g}_k^H\mathbf{F}^{(n+1)}\mathbf{h}_jt_j^{(n)}-\alpha_j\Big|^2
+\gamma\sigma^2_b\|r_k\mathbf{g}_k^H\mathbf{F}^{(n+1)}\|_2^2+\sigma_k^2|r_k|^2\Bigg), \label{problemq}
\\
\{t_k^{(n+1)}\}=\mathop{\mathrm{arg~min}}_{\substack{\{t_k\}}}
&
~\mathop{\mathrm{max}}_{\forall k}~\sum_{j=1}^K\Big|r_k^{(n+1)}\mathbf{g}_k^H\mathbf{F}^{(n+1)}\mathbf{h}_jt_j-\alpha_j\Big|^2 \nonumber\\
\quad\quad\mathrm{s. t.}\quad
&|t_k|^2\leq P_0,\quad k=1,\cdots,K, \label{problemp}
\end{align}
\end{subequations}
where $\{\mathbf{F}^{(n)},t_k^{(n)},r_k^{(n)}\}$ is the solution at the $n$-th outer iteration.
The iterative procedure stops until $n$ reaches the maximum iteration number $n=N_{\mathrm{max}}$.

Problem \eqref{problemF} can be transferred to a convex problem via semidefinite relaxation (SDR) while problems \eqref{problemq} and \eqref{problemp} are convex.
Hence, problems \eqref{problemF}--\eqref{problemp} can all be solved via CVX \cite{opt1}, a Matlab software package for solving convex problems based on IPM.
According to \cite{opt1}, the computational complexity is at least $\mathcal{O}\left(N^{7}\right)$ for solving \eqref{problemF} (the vectorization of $\mathbf{F}$ involves $N^2$ variables) and $\mathcal{O}\left(K^{3.5}\right)$ for solving \eqref{problemq}--\eqref{problemp}.
For large $N$ and $K$, this method is not desirable.
In the following, a new algorithm termed PAM, which decomposes \eqref{problemF}--\eqref{problemp} into smaller subproblems that are either solved by gradient updates or closed-form updates, is proposed for achieving both excellent performance and significantly lower computational complexities.

\subsection{Inner-Layer Iteration}

\subsubsection{Optimization of $\mathbf{F}$}

First, vectorization of $\mathbf{F}$ is given as $\mathbf{f}=\mathrm{vec}\left(\mathbf{F}\right)\in\mathbb{C}^{N^2\times 1}$.
Applying $\mathrm{Tr}\left(\mathbf{A}\mathbf{X}\mathbf{B}\mathbf{X}^T\right)=\mathrm{vec}(\mathbf{X})^T\left(\mathbf{B}^T\otimes\mathbf{A}\right)\mathrm{vec}(\mathbf{X})$ \cite{vec},
we have
\begin{align}
r_k^{(n)}\mathbf{g}_k^H\mathbf{F}\mathbf{h}_jt_j^{(n)}&=
r_k^{(n)}t_j^{(n)}\left(\mathbf{h}_j^T\otimes\mathbf{g}_k^H\right)
\mathrm{vec}\left(\mathbf{F}\right)
\nonumber\\
&
=(\mathbf{a}_{k,j}^{(n)})^H\mathbf{f},
\\
\sigma^2_b\|r_k^{(n)}\mathbf{g}_k^H\mathbf{F}\|_2^2&=|r_k^{(n)}|^2\mathrm{Tr}\left(\mathbf{g}_k\mathbf{g}_k^H\mathbf{F}\mathbf{I}_N\mathbf{F}^H\right)
\nonumber\\
&
=\mathbf{f}^H\mathbf{G}_k^{(n)}\mathbf{f},
\end{align}
where $ \mathbf{a}_{k,j}^{(n)}=\left[r_k^{(n)}t_j^{(n)}\left(\mathbf{h}_j^T\otimes\mathbf{g}_k^H\right)\right]^H$ and $\mathbf{G}_k^{(n)}=\sigma^2_b|r_k^{(n)}|^2\mathbf{I}_N\otimes\left(\mathbf{g}_k\mathbf{g}_k^H\right)$.
Problem \eqref{problemF} is thus re-formulated as
\begin{align}
\mathcal{P}_F:\mathop{\mathrm{min}}_{\substack{\mathbf{f}}}
\quad&\mathop{\mathrm{max}}_{k=1,\cdots,K}~\left(
\sum_{j=1}^K\Big|(\mathbf{a}_{k,j}^{(n)})^H\mathbf{f}-\alpha_j\Big|^2+
\mathbf{f}^H\mathbf{G}_k^{(n)}\mathbf{f}\right)\nonumber\\
\quad\quad\quad\mathrm{s. t.}\quad
&|f_{l}|=1,\quad l=1,\cdots,N^2. \label{PF}
\end{align}

Next, to handle the nonseparable objective function, variable splitting of $\mathbf{f}$ is proposed such that $\mathbf{f}=\mathbf{u}_1=\cdots=\mathbf{u}_K$, where $\{\mathbf{u}_k\}$ are auxiliary variables.
Moreover, to handle the unit-modulus constraints, another auxilliary variable $\mathbf{z}=\mathbf{f}$ is introduced.
For all the newly introduced equality constraints, they can be transformed into quadratic penalties in the objective function \cite{penalty}.
As a result, $\mathcal{P}_F$ is transformed into
\begin{align}
\mathop{\mathrm{min}}_{\substack{\mathbf{f},\mathbf{z},\{\mathbf{u}_{k}\}}}
\quad&
\rho\left(\frac{1}{K}\sum_{j=1}^K\|\mathbf{u}_{j}-\mathbf{f}\|_2^2+\|\mathbf{z}-\mathbf{f}\|_2^2\right)
+
\mathop{\mathrm{max}}_{\forall k}~\left(
\sum_{j=1}^K\Big|(\mathbf{a}_{k,j}^{(n)})^H\mathbf{u}_{k}-\alpha_j\Big|^2+
\mathbf{u}_{k}^H\mathbf{G}_k^{(n)}\mathbf{u}_{k}\right)
\nonumber\\
\mathrm{s. t.}\quad
&|z_{l}|=1,\quad l=1,\cdots,N^2, \label{PFa}
\end{align}
where $\rho$ is a tuning parameter.
It can be proved that $\mathcal{P}_F$ and \eqref{PFa} are equivalent problems as $\rho\rightarrow+\infty$ \cite{penalty}.
However, this case also leads to the gradient norm of the objective function of \eqref{PFa} being infinite, making \eqref{PFa} difficult to solve.
Therefore, $\rho$ controls the tradeoff between approximation error and difficulty in solving \eqref{PFa}.

We address \eqref{PFa} using alternating minimization, in which the cost function is iteratively minimized with respect to one variable whereas the others are fixed.
Starting with an initial $\mathbf{f}^{(0)}=\mathbf{z}^{(0)}=\mathbf{u}_{k}^{(0)}=\mathrm{vec}\left(\mathbf{F}^{(n)}\right)$, the whole process consists of iteratively solving
\begin{subequations}
\begin{align}
\mathbf{u}_{k}^{(m+1)}=\mathop{\mathrm{arg~min}}_{\substack{\mathbf{u}_{k}}}
\quad&\sum_{j=1}^K\Big|(\mathbf{a}_{k,j}^{(n)})^H\mathbf{u}_{k}-\alpha_j\Big|^2+
\mathbf{u}_{k}^H\mathbf{G}_k^{(n)}\mathbf{u}_{k}
+\frac{\rho}{K}\|\mathbf{u}_{k}-\mathbf{f}^{(m)}\|_2^2,\quad \forall k, \label{PFa1}
\\
\mathbf{f}^{(m+1)}=\mathop{\mathrm{arg~min}}_{\substack{\mathbf{f}}}
\quad&\frac{1}{K}\sum_{j=1}^K\|\mathbf{u}_{j}^{(m+1)}-\mathbf{f}\|_2^2+\|\mathbf{z}^{(m)}-\mathbf{f}\|_2^2, \label{PFa2}
\\
\mathbf{z}^{(m+1)}=\mathop{\mathrm{arg~min}}_{\substack{|z_{l}|=1, \forall l}}
\quad&\rho\|\mathbf{z}-\mathbf{f}^{(m+1)}\|_2^2, \label{PFa3}
\end{align}
\end{subequations}
where $m$ is the inner iteration index.
It can be verified that the objective function of \eqref{PFa} is strongly convex.
Therefore, despite the non-differentiability of the objective, the alternating minimization \eqref{PFa1}--\eqref{PFa3} is guaranteed to converge to a stationary point of \eqref{PFa} \cite{beck}.
The iterative procedure stops until $m$ reaches the maximum iteration number $m=M_{\mathrm{max}}$.

The remaining question is how to solve \eqref{PFa1}--\eqref{PFa3} optimally.
We notice that problems \eqref{PFa1} and \eqref{PFa2} are standard least squares problems, thus their solutions are given by the following closed-form expressions
\begin{align}
\mathbf{u}_{k}^{(m+1)}&=\left(
\sum_{j=1}^K\mathbf{a}_{k,j}^{(n)}(\mathbf{a}_{k,j}^{(n)})^H+\mathbf{G}_k^{(n)}+\rho\mathbf{I}
\right)^{-1}
\times
\left(\sum_{j=1}^K\alpha_j\mathbf{a}_{k,j}^{(n)}+\frac{\rho}{K}\mathbf{f}^{(m)}\right), \label{ukm}
\\
\mathbf{f}^{(m+1)}&=\frac{1}{2}\left(\frac{1}{K}\sum_{j=1}^K\mathbf{u}_{j}^{(m+1)}+\mathbf{z}^{(m)}\right), \label{fm}
\end{align}
respectively.
On the other hand, problem \eqref{PFa3} is the projection of $\mathbf{f}^{(m+1)}$ onto unit-modulus constraint and the optimal solution is simply
\begin{align}
\mathbf{z}^{(m+1)}=\mathrm{exp}\left(\mathrm{j}\angle\mathbf{f}^{(m+1)}\right). \label{zm}
\end{align}

\subsubsection{Optimization of $\{r_k\}$}
The problem of optimizing $\{r_k\}$ in \eqref{problemq}  is also a least squares problem.
The optimal solution is found by setting the derivative $\partial \mathbb{MSE}_k/\partial \mathrm{conj}(r_k)$ to zero:
\begin{align}
\frac{\partial \mathbb{MSE}_k}{\partial \mathrm{conj}(r_k)}
&=\gamma
\sum_{j=1}^K\left(r_k\mathbf{g}_k^H\mathbf{F}^{(n+1)}\mathbf{h}_jt_j^{(n)}-\alpha_j\right)
\mathrm{conj}\left(\mathbf{g}_k^H\mathbf{F}^{(n+1)}\mathbf{h}_jt_j^{(n)}\right)
\nonumber\\
&\quad{}
+\gamma\sigma^2_b\|\mathbf{g}_k^H\mathbf{F}^{(n+1)}\|_2^2r_k
+\sigma_k^2r_k=0,
\end{align}
which yields
\begin{align}
r_k^{(n+1)}=&
\frac{\sum_{j}\alpha_j\mathrm{conj}\left(\mathbf{g}_k^H\mathbf{F}^{(n+1)}\mathbf{h}_jt_j^{(n)}\right)}
{\sum_{j}|\mathbf{g}_k^H\mathbf{F}^{(n+1)}\mathbf{h}_jt_j^{(n)}|^2
+\sigma^2_b\|\mathbf{g}_k^H\mathbf{F}^{(n+1)}\|_2^2
+\sigma_k^2/\gamma}
. \label{qk}
\end{align}

\subsubsection{Optimization of $\{t_k\}$}
The objective function of problem \eqref{problemp} is not separable.
Following similar variable splitting procedure as in \eqref{PFa}, we introduce auxiliary variables $\{\xi_{k,j}=t_k, \forall k,j\}$ and add quadratic penalty $\rho\sum_{k=1}^K\sum_{j=1}^K|\xi_{k,j}-t_k|^2$ to the objective function.
The problem \eqref{problemp} is transformed into
\begin{align}
\mathcal{P}_t:\mathop{\mathrm{min}}_{\substack{\{t_k,\xi_{k,j}\}}}
\quad&\rho\sum_{k=1}^K\sum_{j=1}^K|\xi_{k,j}-t_j|^2
+
\mathop{\mathrm{max}}_{\forall k}~\sum_{j=1}^K\Big|r_k^{(n+1)}\mathbf{g}_k^H\mathbf{F}^{(n+1)}\mathbf{h}_j\xi_{k,j}-\alpha_j\Big|^2 \nonumber\\
\quad\quad\quad\mathrm{s. t.}\quad
&|t_k|^2\leq P_0,\quad k=1,\cdots,K,
 \label{Pk}
\end{align}
and variables $t_k$ and $\xi_{k,j}$ can be optimized iteratively.
In particular, starting from $t_k^{(0)}=\xi_{k,j}^{(0)}=t_k^{(n)}$, the solutions of $\xi_{k,j}$ and $t_k$ at the $q$-th iteration are given by
\begin{align}
\xi_{k,j}^{(q+1)} =&\mathop{\mathrm{arg~min}}_{\substack{\xi_{k,j}}}~
\Big|r_k^{(n+1)}\mathbf{g}_k^H\mathbf{F}^{(n+1)}\mathbf{h}_j\xi_{k,j}-\alpha_j\Big|^2
+\rho|\xi_{k,j}-t_j^{(q)}|^2,\quad\forall k,j, \label{xikj}
\\
t_{j}^{(q+1)} = &\mathop{\mathrm{arg~min}}_{\substack{|t_j|^2\leq P_0}}~
\rho\sum_{k=1}^K|\xi_{k,j}^{(q+1)}-t_j|^2,\quad \forall j. \label{pj}
\end{align}
Problem \eqref{xikj} is a least squares problem and \eqref{pj} is a quadratic problem with only one constraint.
They can be solved optimally based on Karush-Kuhn-Tucker (KKT) conditions and the solutions are given by
\begin{align}
\xi_{k,j}^{(q+1)} &= \frac{\mathrm{conj}(r_k^{(n+1)}\mathbf{g}_k^H\mathbf{F}\mathbf{h}_j)\alpha_j+\rho t_j^{(q)}}{|r_k^{(n+1)}\mathbf{g}_k^H\mathbf{F}^{(n+1)}\mathbf{h}_j|^2+\rho},\quad\forall k,j, \label{xikm+1}
\\
t_{j}^{(q+1)} &= \sqrt{\mathrm{min}\left(P_0, 1\right)}\ \frac{\frac{1}{K}\sum_{k=1}^K\xi_{k,j}^{(q+1)}}{\left|\frac{1}{K}\sum_{j=1}^K\xi_{k,j}^{(q+1)}\right|}, \quad\forall j.\label{pkm+1}
\end{align}
The iterative procedure stops until $q$ reaches the maximum iteration number $q=Q_{\mathrm{max}}$.

\subsection{Summary and Complexity Analysis of PAM}

In summary, the complete PAM algorithm for solving problem $\mathcal{P}$ with a fully-connected structure consists of two layers of iterations.
Let $N_{\rm{max}}$ and $M_{\rm{max}}$ denote the maximum number of iterations for outer and inner layers, respectively.
In the outer layer, the PAM optimizes $\mathbf{F}$, $\{r_k\}$ and $\{t_k\}$ alternatively in each of the $N_{\rm{max}}$ iterations.
In the inner layer, $\mathbf{F}$ is obtained via computing \eqref{ukm}--\eqref{zm} for $M_{\rm{max}}$ iterations, $\{r_k\}$ is obtained via computing \eqref{qk}, and $\{t_k\}$ is obtained via computing \eqref{xikm+1}--\eqref{pkm+1} for $Q_{\rm{max}}$ iterations.
The computational complexities for these equations are $\mathcal{O}(KN^2)$, $\mathcal{O}(KN^2)$, $\mathcal{O}(N^2)$, $\mathcal{O}(KN^2)$, $\mathcal{O}(K^2)$, $\mathcal{O}(K^2)$ for \eqref{ukm}, \eqref{fm}, \eqref{zm}, \eqref{qk}, \eqref{xikm+1}, \eqref{pkm+1}, respectively.
Since the computation is dominant by $\mathcal{O}(KN^2)$, the total computational complexity of PAM is $\mathcal{O}(N_{\rm{max}}Q_{\rm{max}}KN^2)$.

\section{Accelerated Gradient Projection for Partially-Connected UMAirComp}

In practice, it is possible that there are a large number of antennas at the edge server.
In such a case, a smaller number of phase shifters than $N^2$ at the edge server is desirable.
To this end, this section proposes an accelerated gradient projection method for partially-connected UMAirComp, which only needs $2N$ phase shifters.

With a partially-connected structure as illustrated in Fig.~2b, the feasible set $\mathcal{F}$ equals the second line of \eqref{setF}.
Since the rank of $\mathbf{F}$ is $1$, we can apply rank-one decomposition on $\mathbf{F}$ which yields $\mathbf{F}=\mathbf{v}\mathbf{w}^H$.
Then, we adopt the following approximations to $\mathcal{P}$: 1) Set $r_k=1/(\mathbf{g}_k^H\mathbf{v})$ and $t_k=1/(\mathbf{h}_k^H\mathbf{w})$; 2) Relax
$\{|F_{l,l'}|=1|\forall l,l'\}$ into $\|\mathbf{F}\|_2^2\leq N^2$.
After the above steps, problem $\mathcal{P}$ is simplified into a bilevel form:
\begin{subequations}
\begin{align}
\mathcal{Q}:
\mathop{\mathrm{min}}_{\substack{\mathbf{v}}}
\quad&
\gamma\sigma^2_b\|\mathbf{w}\|_2^2+\mathop{\mathrm{max}}_{k=1,\cdots,K}\,\frac{\sigma_k^2}{|\mathbf{g}_k^H\mathbf{v}|^2}
\\
 \quad\mathrm{s. t.}\quad & \frac{N^2}{\|\mathbf{v}\|_2^2}\geq
\mathop{\mathrm{min}}_{\mathbf{w}}\,\left\{\|\mathbf{w}\|_2^2:|\mathbf{w}^H\mathbf{h}_k|^2\geq\frac{\alpha_k^2}{P_0},\quad\forall k\right\}. \label{Q1}
\end{align}
\end{subequations}
Note that the adopted approximations above $\mathcal{Q}$ would lead to performance loss.
However, the solution obtained by solving $\mathcal{Q}$ is asymptotically optimal to $\mathcal{P}$ as $\sigma_k^2,\sigma_b^2\rightarrow 0$.
This is because $\gamma\sum_{j=1}^K|r_k\mathbf{g}_k^H\mathbf{F}\mathbf{h}_jt_j-\alpha_j|^2=0$ with $\mathbf{F}=\mathbf{v}\mathbf{w}^H$, $r_k=1/(\mathbf{g}_k^H\mathbf{v})$, and $t_k=\alpha_k/(\mathbf{h}_k^H\mathbf{w})$.
Consequently, as noise powers $\sigma^2_b,\sigma^2_k\rightarrow 0$, the objective function of $\mathcal{P}$ with $\{\mathbf{F}=\mathbf{v}\mathbf{w}^H,r_k=1/(\mathbf{g}_k^H\mathbf{v}), t_k=\alpha_k/(\mathbf{h}_k^H\mathbf{w})\}$ becomes zero, meaning that this solution is optimal to $\mathcal{P}$ in the asymptotic case.

Since the right hand side of the constraint in \eqref{Q1} is a quadratic optimization problem, it can be solved by the accelerated random coordinate descent method with a complexity of $\mathcal{O}(KN)$ \cite{arcd}.
Denoting the solution of $\mathbf{w}$ as $\mathbf{w}=\mathbf{w}^\diamond$, problem $\mathcal{Q}$ is reduced to
\begin{align}
\mathcal{Q}_1:\mathop{\mathrm{min}}_{\substack{\mathbf{v}}}
\quad&\mathop{\mathrm{max}}_{k=1,\cdots,K}\,\frac{\sigma_k^2}{|\mathbf{g}_k^H\mathbf{v}|^2}  \quad \mathrm{s. t.}\quad \|\mathbf{v}\|_2^2\leq \beta,
\end{align}
where $\beta=\frac{N^2}{\|\mathbf{w}^\diamond\|_2^2}$ and we have removed the term $\gamma\sigma^2_b\|\mathbf{w}^\diamond\|_2^2$ in the objective function (as it is a constant now).
In the following, we propose an efficient fixed point method for solving problem $\mathcal{Q}_1$.

\subsection{Fixed-Point Iteration}

Solving problem $\mathcal{Q}_1$ is challenging due to 1) the large dimension of variables, and 2) the large number of elements inside the maximum operator.
To this end, we first re-write $\mathcal{Q}_1$ as a bilevel problem
\begin{align}
\mathop{\mathrm{min}}_{\substack{\|\mathbf{v}\|_2^2\leq \beta}}
\,\mathop{\mathrm{max}}_{k=1,\cdots,K}\,\frac{\sigma_k^2}{|\mathbf{g}_k^H\mathbf{v}|^2}
&\Longleftrightarrow
\mathop{\mathrm{min}}_{\substack{\|\mathbf{v}\|_2^2\leq \beta}}
\,\mathop{\mathrm{max}}_{k=1,\cdots,K}\,-\frac{|\mathbf{g}_k^H\mathbf{v}|^2}{\sigma_k^2}
\nonumber\\
&
\Longleftrightarrow
\mathop{\mathrm{min}}_{\substack{\|\mathbf{v}\|_2^2\leq \beta}}
\,\mathop{\mathrm{max}}_{\mathbf{b}\in\Delta}\,\underbrace{-\mathop{\sum}_{k=1}^K\frac{b_k|\mathbf{g}_k^H\mathbf{v}|^2}{\sigma_k^2}}_{:=h(\mathbf{v},\mathbf{b})}, \label{P2'}
\end{align}
where $\Delta=\{\mathbf{b}|\mathbf{b}\succeq\mathbf{0},\,\mathbf{1}^T\mathbf{b}=1\}$ and the last step is used to smooth the objective function via introducing one more auxiliary optimization variable $\mathbf{b}$ \cite{minimax}.
Then, we have the following conclusion on the Karush-Kuhn-Tucker solution to \eqref{P2'}, which also holds for $\mathcal{Q}_1$.
\begin{lemma}
Let
\begin{align}
&U(\mathbf{v}')=\frac{\sqrt{\beta}\mathbf{C}(\mathbf{v}')\mathop{\mathrm{arg~min}}_{\mathbf{b}\in\Delta}\,
\Phi\left(\mathbf{v}',\mathbf{b}\right)}{\|\mathbf{C}(\mathbf{v}')\mathop{\mathrm{arg~min}}_{\mathbf{b}\in\Delta}\,
\Phi\left(\mathbf{v}', \mathbf{b}\right)\|_2}, \label{mdg2}
\end{align}
where
\begin{align}
\Phi\left(\mathbf{v}',\mathbf{b}\right)&=
2\sqrt{\beta}\|\mathbf{C}(\mathbf{v}')\mathbf{b}\|_2-\left[\mathbf{q}(\mathbf{v}')\right]^T\mathbf{b}, \label{Phi}
\\
\mathbf{C}(\mathbf{v}')&=\left[\frac{\mathbf{g}_{1}\mathbf{g}^H_{1}\mathbf{v}'}{\sigma_1^2},\cdots,\frac{\mathbf{g}_{K}\mathbf{g}^H_{K}\mathbf{v}'}{\sigma_K^2}\right]\in\mathbb{C}^{N\times K}, \label{Cn}
\\
\mathbf{q}(\mathbf{v}')&=\left[\frac{|\mathbf{g}_{1}^H\mathbf{v}'|^2}{\sigma_1^2},\cdots,\frac{|\mathbf{g}_{K}^H\mathbf{v}'|^2}{\sigma_K^2}\right]^T
\in\mathbb{C}^{K\times 1}. \label{dn}
\end{align}
Then with any feasible $\mathbf{v}^{(0)}$ and fixed point iteration $\mathbf{v}^{(n+1)}\leftarrow U(\mathbf{v}^{(n)})$, every limit point $\mathbf{v}^\diamond$ of the sequence $\{\mathbf{v}^{(0)},\mathbf{v}^{(1)},...\}$ is a Karush-Kuhn-Tucker solution to problem \eqref{P2'}.
\end{lemma}
\begin{proof}
It can be proved that $U(\mathbf{v}')$ is the optimal solution to $\mathop{\mathrm{min}}_{\|\mathbf{v}\|_2^2\leq \beta}\mathop{\mathrm{max}}_{\mathbf{b}\in\Delta}\,g(\mathbf{v},\mathbf{v}',\mathbf{b})$, where $g(\mathbf{v},\mathbf{v}',\mathbf{b})$ is a surrogate function satisfying $g(\mathbf{v},\mathbf{v}',\mathbf{b})\geq h(\mathbf{v},\mathbf{b})$, $g(\mathbf{v}',\mathbf{v}',\mathbf{b})=h(\mathbf{v}',\mathbf{b})$, and $\nabla g(\mathbf{v}',\mathbf{v}',\mathbf{b})\geq \nabla h(\mathbf{v}',\mathbf{b})$.
According to \cite{mm} and the properties of $g(\mathbf{v},\mathbf{v}',\mathbf{b})$, every limit point of the sequence $(\mathbf{v}^{(0)},\mathbf{v}^{(1)},\cdots)$ generated by $\mathbf{v}^{(n+1)}\leftarrow U(\mathbf{v}^{(n)})$ and a feasible $\mathbf{v}^{(0)}$ is the KKT solution to \eqref{P2'}.
More details can be found in Appendix C.
\end{proof}

Although Lemma 1 reveals the solution structure of \eqref{P2'}, computation of $U(\mathbf{v}')$ is not straightforward as it involves another optimization problem of $\mathbf{b}$, which should also be solved with low computation costs.
In the following, we will solve the optimization problem of $\mathbf{b}$ via the methods of smoothing and acceleration.

\subsection{Optimization of $\mathbf{b}$ via Smoothing and Acceleration}

In order to compute $U(\mathbf{v}')$, a necessary step is to find the optimal vector
\begin{align}
\mathbf{b}^*=\mathop{\mathrm{arg~min}}_{\mathbf{b}\in\Delta}\Phi\left(\mathbf{b}\right), \label{mdg1}
\end{align}
where we have omitted the symbol $\mathbf{v}'$ in \eqref{Phi} since $\mathbf{v}'$ is a known and fixed vector in each iteration.
Notice that the gradient of the objective
\begin{align}\label{gradient1}
\nabla_{\mathbf{b}}\Phi(\mathbf{b})=
&
\frac{2\sqrt{\beta}\,\mathrm{Re}\left(\mathbf{C}^H\mathbf{C}\mathbf{b}\right)}{\|\mathbf{C}\mathbf{b}\|_2}
-\mathbf{q}
\end{align}
is unbounded when $\|\mathbf{C}\mathbf{b}\|_2\rightarrow0$, which happens if $\mathbf{b}\in\mathrm{Null}(\mathbf{C})$.
Therefore, it is nontrivial to apply first-order method to problem \eqref{mdg1}.
To avoid the unbounded gradients, we adopt the smoothing technique \cite{smoothing} to replace $\Phi(\mathbf{b})$ in \eqref{mdg1} with
\begin{align}
\Xi(\mathbf{b})
=
&
2\sqrt{\beta}\times
\sqrt{\phi^2+
\big|\big|\mathbf{C}\mathbf{b}\big|\big|_2^2}
-\mathbf{q}^T\mathbf{b}, \label{Xi}
\end{align}
where the tuning parameter $\phi\geq 0$ such that $\Xi(\mathbf{b})=\Phi(\mathbf{b})$ for $\phi=0$.
Then problem \eqref{mdg1} can be approximated by
\begin{align}
&\mathcal{Q}_2:\mathop{\mathrm{min}}_{\substack{\mathbf{b}\in \Delta}}
~\Xi\left(\mathbf{b}\right).
\end{align}
In the following, we first elaborate the optimal solution of problem $\mathcal{Q}_2$, and then establish the relation between the solutions of problems $\mathcal{Q}_2$ and \eqref{mdg1}.
First of all, we have the following lemma on the objective of problem $\mathcal{Q}_2$.

\begin{lemma}
$\Xi(\mathbf{b})$ is Lipschitz smooth for $\mathbf{b}\in\Delta$, with the Lipschitz constant of gradient
\begin{align}\label{Lip1}
L_{\Xi}(\phi)=\frac{2\sqrt{\beta}\,\lambda_{\mathrm{max}}\left[\mathrm{Re}\left(\mathbf{C}^H\mathbf{C}\right)\right]}
{\sqrt{\phi^2+\lambda_{\mathrm{min}}\left(\mathbf{C}^H\mathbf{C}\right)/K}}.
\end{align}
\end{lemma}
\begin{proof}
The proof is based on bounding the eigenvalue of the Hessian matrix of $\Xi(\mathbf{b})$.
More details can be found in Appendix D.
\end{proof}

Lemma 2 shows that $\mathcal{Q}_2$ is a Lipschitz smooth problem.
As a result, the acceleration method \cite{acceleration1,acceleration3} can be adopted to optimally solve $\mathcal{Q}_2$ iteratively.
The algorithm is summarized in Theorem 3.

\begin{theorem}
Let $\mathbf{b}(0)\in\Delta$ and
\begin{align}\label{accelerated}
&\mathbf{b}(m+1)
=\Pi_{\Delta}\Big[
\bm{\rho}(m)-\frac{1}{L_{\Xi}(\phi)}\nabla_{\mathbf{b}}\Xi(\mathbf{b})\Big|_{\mathbf{b}=\bm{\rho}(m)}\Big],
\end{align}
where $m$ is the iteration index, $\Pi_{\Delta}$ is the projection onto set $\Delta$, $L_{\Xi}(\phi)$ is defined in \eqref{Lip1}, and
\begin{align}\label{gradient2}
\nabla_{\mathbf{b}}\Xi(\mathbf{b})=
&
\mathbf{q}-2\sqrt{\beta}\times\frac{\mathrm{Re}\left(\mathbf{C}^H\mathbf{C}\mathbf{b}\right)}{\sqrt{\phi^2+\|\mathbf{C}\mathbf{b}\|_2^2}},
\\
\bm{\rho}(m)=&\mathbf{b}(m)+\frac{c(m-1)-1}{c(m)}\left(\mathbf{b}(m)-\mathbf{b}(m-1)\right), \label{rho}
\\
c(m)=&\frac{1}{2}\left(1+\sqrt{1+4\left(c(m-1)\right)^2}\right),\quad c(0)=1. \label{cm}
\end{align}
Then the sequence computed from \eqref{accelerated}--\eqref{cm} converges to the optimal solution of $\mathcal{Q}_2$ with an iteration complexity
$\mathcal{O}\left(\sqrt{L_{\Xi}(\phi)/\epsilon}\right)$, where $\epsilon$ is the target accuracy.
\end{theorem}
\begin{proof}
It can be proved by following a similar approach in \cite[Theorem 4.4]{acceleration3}.
\end{proof}
Notice that the iteration complexity touches the lower bound derived in \cite[Theorem 2.1.6]{yurii}.
The computation of the projection $\Pi_{\Delta}(\mathbf{u})$ given the input vector $\mathbf{u}$ is summarized in Lemma 3.
\begin{lemma}
Let $\mathbf{u}'=\mathrm{sort}(\mathbf{u})$, where the function $\mathrm{sort}$ permutes the elements of $\mathbf{u}$ in a descent order such that $u_1'\geq \cdots \geq u_K'$ and
$
\delta=\mathop{\mathrm{max}}_{x\in\{1,\cdots,K\}}~\left\{x:
\frac{\sum_{l=1}^xu_l'-1}{x}<u_x'
\right\}$.
Then
\begin{align}
&
\Pi_{\Delta}(\mathbf{u})=
\left(\mathbf{u}-\frac{\mathop{\sum}_{l=1}^\delta u_l'-1}{\delta}\right)^+.
\end{align}
\end{lemma}
\begin{proof}
Please refer to \cite[Proposition 2.2]{proj}.
\end{proof}

Finally, we have the following conclusion on the relation between the solutions to problems $\mathcal{Q}_2$ and \eqref{mdg1}.

\begin{theorem}
(i) If $\mathrm{Rank}\left([\mathbf{g}_1,\cdots,\mathbf{g}_K]\right)=K$ (thus $L_{\Xi}(0)<+\infty$), the optimal solution to problem $\mathcal{Q}_2$ is optimal to problem \eqref{mdg1} by setting $\phi=0$. (ii) If $\mathrm{Rank}\left([\mathbf{g}_1,\cdots,\mathbf{g}_K]\right)\neq K$, then $L_{\Xi}(0)=+\infty$, and $L_{\Xi}(\phi)<+\infty$ if $\phi>0$. (iii) For all $\mathbf{b}'\in\Delta$ with $\Xi(\mathbf{b}')-\Xi(\mathbf{b}^\diamond)\leq \epsilon$,
\begin{align}
&\Phi(\mathbf{b}')-\Phi(\mathbf{b}^*)\leq 2\sqrt{\beta}\,\phi+\epsilon,
\end{align}
where $\mathbf{b}^\diamond$ and $\mathbf{b}^*$ denote the optimal solutions to $\mathcal{Q}_2$ and \eqref{mdg1}, respectively.
\end{theorem}
\begin{proof}
Please refer to Appendix E.
\end{proof}

Part (i) of Theorem 4 indicates that we can always set $\phi=0$ if the user channels are independent.
In this case, $\Phi(\mathbf{b})=\Xi(\mathbf{b})$, which means that the optimal solution to problem $\mathcal{Q}_2$ is the same as that of \eqref{mdg1}.
On the other hand, part (ii) of Theorem 4 indicates that if the user channels are correlated, we must choose $\phi>0$, and the conversion from \eqref{mdg1} to $\mathcal{Q}_2$ would lead to approximation error.
However, this error is controllable by choosing a small $2\sqrt{\beta}\,\phi$ according to part (iii) of Theorem 4 (e.g., with $\phi=0.1/(2\sqrt{\beta})$, the approximation error is at most $2\sqrt{\beta}\,\phi=0.1$).

\subsection{Summary and Complexity Analysis of AGP}

For the proposed AGP algorithm, the accelerated random coordinate descent is first used to compute $\mathbf{w}^\diamond$ for problem on the right hand side of the constraint in \eqref{Q1}, which requires a complexity of $\mathcal{O}\left(KN\right)$.
To optimize $\mathbf{v}^\diamond$, in each fixed-point iteration, the terms $\mathbf{C}$ in \eqref{Cn} and $\mathbf{q}$ in \eqref{dn} are computed with a complexity of $\mathcal{O}(KN)$,
followed by the iterative calculation of variable $\mathbf{b}$ in \eqref{mdg1} with equations \eqref{accelerated}--\eqref{cm}, which involves a complexity of $\mathcal{O}(KN)$ for gradient computation.
Therefore, the overall complexity of AGP for solving $\mathcal{Q}$ is $\mathcal{O}\left(KN\right)$.
Notice that with the obtained $\mathbf{w}^\diamond$ and $\mathbf{v}^\diamond$, we need to recover $\{\mathbf{F}^\star,r_k^\star,t_k^\star\}$.
To satisfy the unit-modulus constraints, $\mathbf{w}^\diamond$ and $\mathbf{v}^\diamond$ are projected onto the unit-modulus constraints as $\mathbf{w}^\star=\mathrm{exp}(\mathrm{j}\angle\mathbf{w}^\diamond)$ and $\mathbf{v}^\star=\mathrm{exp}(\mathrm{j}\angle\mathbf{v}^\diamond)$.
The recovered phase shift network is $\mathbf{F}^\star=(\mathbf{w}^\star)^H\mathbf{v}^\star$.
With $\mathbf{F}^\star$, $t_k^\star=1/(K(\mathbf{w}^\star)^H\mathbf{h}_k)$ and $r_k^\star$ is computed using \eqref{qk}.

\section{Simulation Results and Discussions}

This section presents simulation results to verify the performance of the proposed scheme.
The pathloss of the user $k$ is set to $\varrho_{k}=-60\,\mathrm{dB}$, and $\mathbf{h}_{k}$ and $\mathbf{g}_{k}$ are generated according to $\mathcal{CN}(\mathbf{0},\varrho_{k}\mathbf{I}_N)$.
It is assumed that the channels during the model uploading and aggregation procedure are static in each federated learning iteration.
On the other hand, different iterations adopt independent channels and noise realizations.
The power scaling factor $\gamma=1$ at the server.
All problems are solved by Matlab R2019a on a desktop with an Intel Core i7-7700 CPU at 3.6\,GHz and 16\,GB RAM.
The Interior point method is implemented using CVX Mosek \cite{opt1}.
Under the above setting, we consider the following benchmark schemes for comparison.
\begin{itemize}
\item[1)]
The optimized user transceiver scheme, which can be viewed as the extension of \cite{air2,air3} to the multi-antenna scenarios.
This scheme optimizes transceiver designs $\{t_k,r_k\}$ but adopts unoptimized beamformer $\mathbf{F}=\mathbf{I}_N$.
The transmitters $\{t_k\}$ and receivers $\{r_k\}$ are optimized iteratively by solving \eqref{problemp} using CVX and \eqref{problemq} using \eqref{qk}, respectively.

\item[2)] The digital beamforming scheme in \cite{air4}, which ignores the unit-modulus constraints on $\mathbf{F}$.
As such, the obtained MSE (training loss, test error) of the digital beamfomring scheme serves as a lower bound on that of the proposed scheme.

\item[3)] The digital beamformer projection scheme, which sets $\mathbf{F}$ as the projection of the digital beamforming design in \cite{air5} onto the unit-modulus constraints.

\item[4)] The UMAirComp scheme with SDR and CVX \cite{opt1}. The analog beamformer $\mathbf{F}$, the transmitters $\{t_k\}$, and the receivers $\{r_k\}$ are optimized iteratively by solving \eqref{problemF} using SDR, \eqref{problemp} using CVX, and \eqref{problemq} using \eqref{qk}, respectively.\footnote{
If the matrix solution to the SDR problem of \eqref{problemF} is not rank-one, we use the principal component of the obtained matrix as the phase shift design.}
\end{itemize}

\subsection{Performance Evaluation of PAM-based and AGP-based UMAirComp}

First, to verify the learning performance of PAM-based UMAirComp, we consider the image classification task based on a convolutional neural network (CNN).
The mixed national institute of standards and technology (MNIST) dataset \cite{MNIST} is used, where $\mathbf{d}_{k,l}^{\mathrm{in}}\in\mathbb{R}^{784\times 1}$ is a gray-scale vector of images and $\mathbf{d}_{k,l}^{\mathrm{out}}\in\mathbb{R}^{10\times 1}$ is a label vector containing only one non-zero element.
The CNN consists of two convolution layers, two max pooling layers, and a fully-connected layer.
The loss function
\begin{align}
&\Theta(\mathbf{d}_{k,l},\mathbf{x}_k)=\left[f_{\mathrm{cnn}}\left(\mathbf{x}_k,\mathbf{d}_{k,l}^{\mathrm{in}}\right)-\mathbf{d}_{k,l}^{\mathrm{out}}\right]^2/2,
\end{align}
where $f_{\mathrm{cnn}}\left(\mathbf{x}_k,\mathbf{d}_{k,l|}^{\mathrm{in}}\right)$ is the softmax output of CNN.
The training step-size is $\varepsilon=1$, the number of local epochs is $4$, and the total number of FL iterations is set to $400$.
For each iteration, the maximum transmit powers at users are $P_0=10~\mathrm{mW}$ (i.e., $10\,\mathrm{dBm}$).
The noise powers at the server and users are set as $-80\,\mathrm{dBm}$, which capture the effects of thermal noise, receiver noise, and interference.
We set $M_{\mathrm{max}}=Q_{\mathrm{max}}=200$ and $N_{\mathrm{max}}=20$ for the inner and outer layer iterations for the PAM algorithm.

Comparison among four benchmark schemes and the proposed scheme is shown in Fig.~4.
Specifically, the axis of the radar map in Fig.~4a ranges from $0$ to $0.3$ for MSE (i.e., the objective function of $\mathcal{P}$), training loss, and test error.
The MSE is obtained by averaging over the $400$ FL iterations and the training loss (test error) is obtained when the entire FL procedure terminates.
Since our goal is to minimize all these metrics concurrently, a smaller area indicates better performance.
Firstly, the optimized user transceiver scheme has the largest area, which can be treated as a worst-case performance bound.
This is because the phase shift network at the server does not align with the users' channels with the optimized user transceiver scheme.
Secondly, the proposed PAM-based UMAirComp scheme achieves a similar size of triangle as that of the digital beamforming scheme.
This is because for federated learning systems, there is no need to decode all the local model parameters and the beam direction can be less accurate.
As such, analog beamforming leads to negligible performance loss compared with the digital beamforming.
Furthermore, the area of the proposed scheme is completely covered by that of the digital beamformer projection scheme and the area reduction of the former comes from the analog beamforming design based on PAM.
Finally, the scheme with SDR and CVX also achieves competing performance.
However, it is at the cost of high computational complexities and execution time, making it not applicable in practice.

On the other hand, Fig.~4b and Fig.~4c show the training loss and the test error versus the number of FL iterations, respectively.
Firstly, due to the high MSE (close to $0.1$), the optimized user transceiver scheme diverges.
Secondly, the FL scheme with digital beamformer projection is competitive compared with the proposed scheme at the beginning, but the training loss does not further decrease with the increasing number of FL iterations.
Indeed, the test error increases after $200$ iterations.
This means that model parameter errors have a stronger impact as the number of FL iterations increases.
This is because deep learning models are usually over parameterized.
Thus, as the training procedure gets closer to convergence, the model parameters would become sparser, which are more sensitive to errors.
Lastly, the proposed scheme achieves a training loss curve and a test error curve close to their associated lower bounds.
This again demonstrates the small performance gap between the analog and the digital beamforming in federated learning systems.

\begin{figure*}[!t]
 \centering
\subfigure[]{\includegraphics[height=65mm]{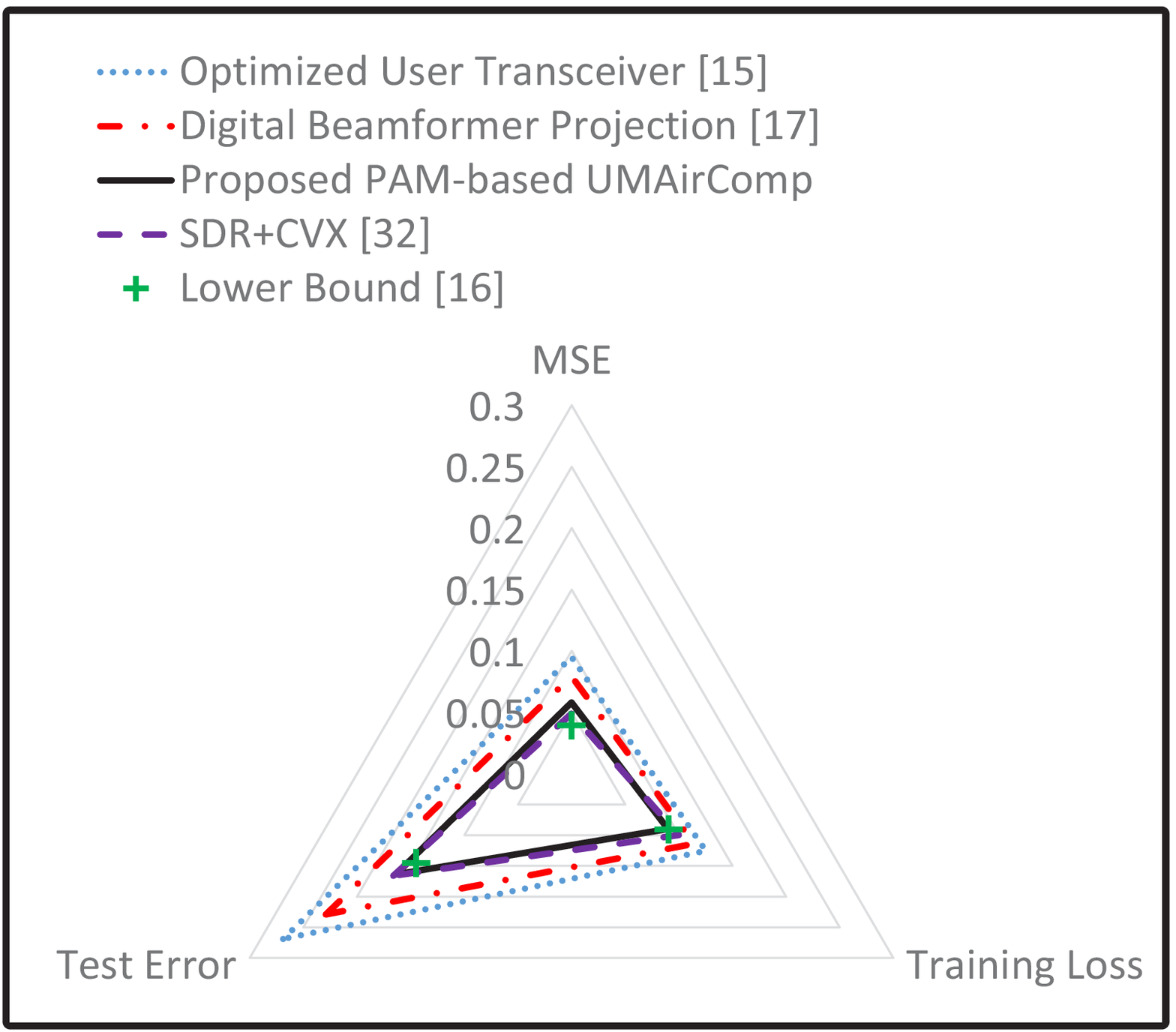}}
\subfigure[]{\includegraphics[height=65mm]{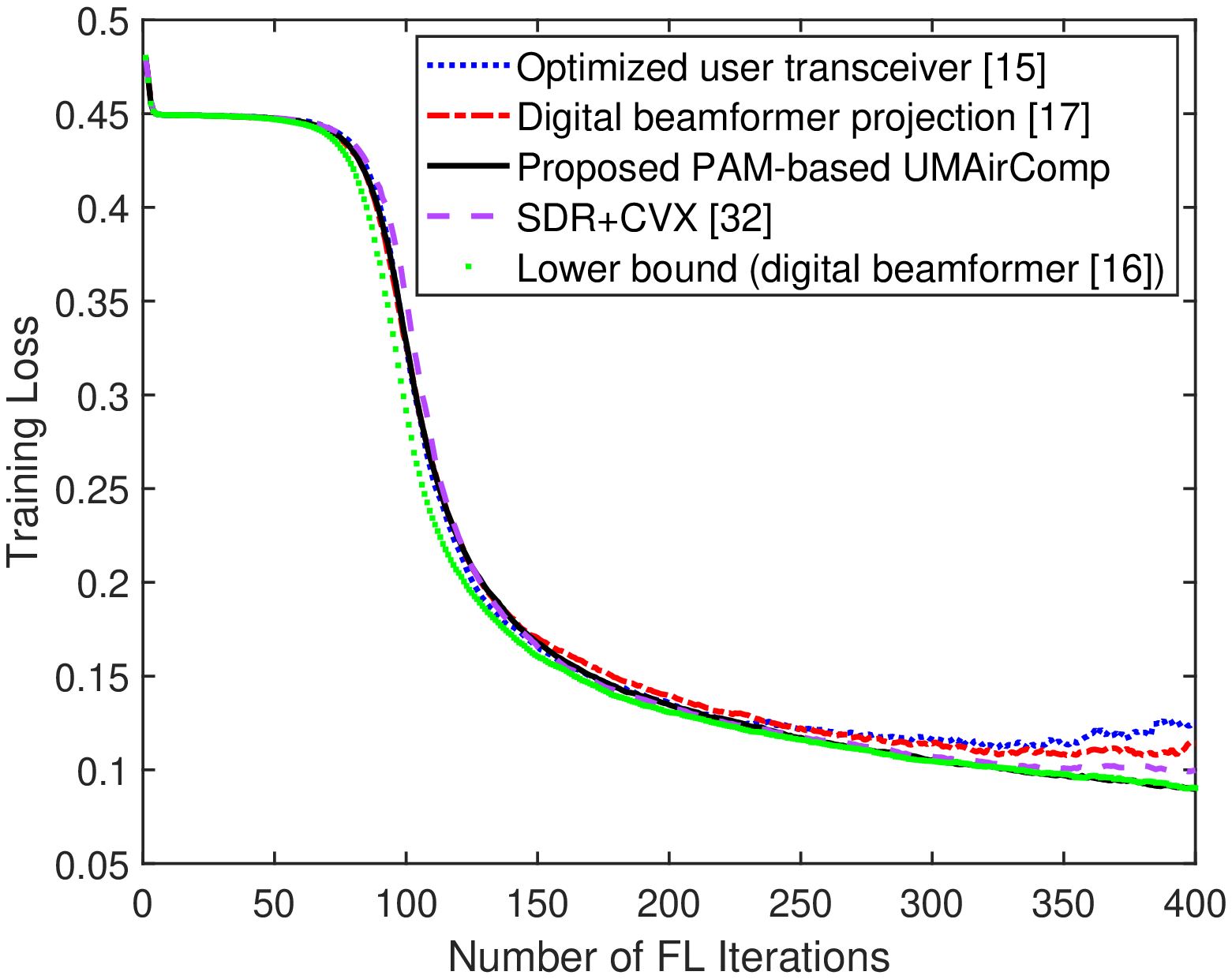}}
\subfigure[]{\includegraphics[height=65mm]{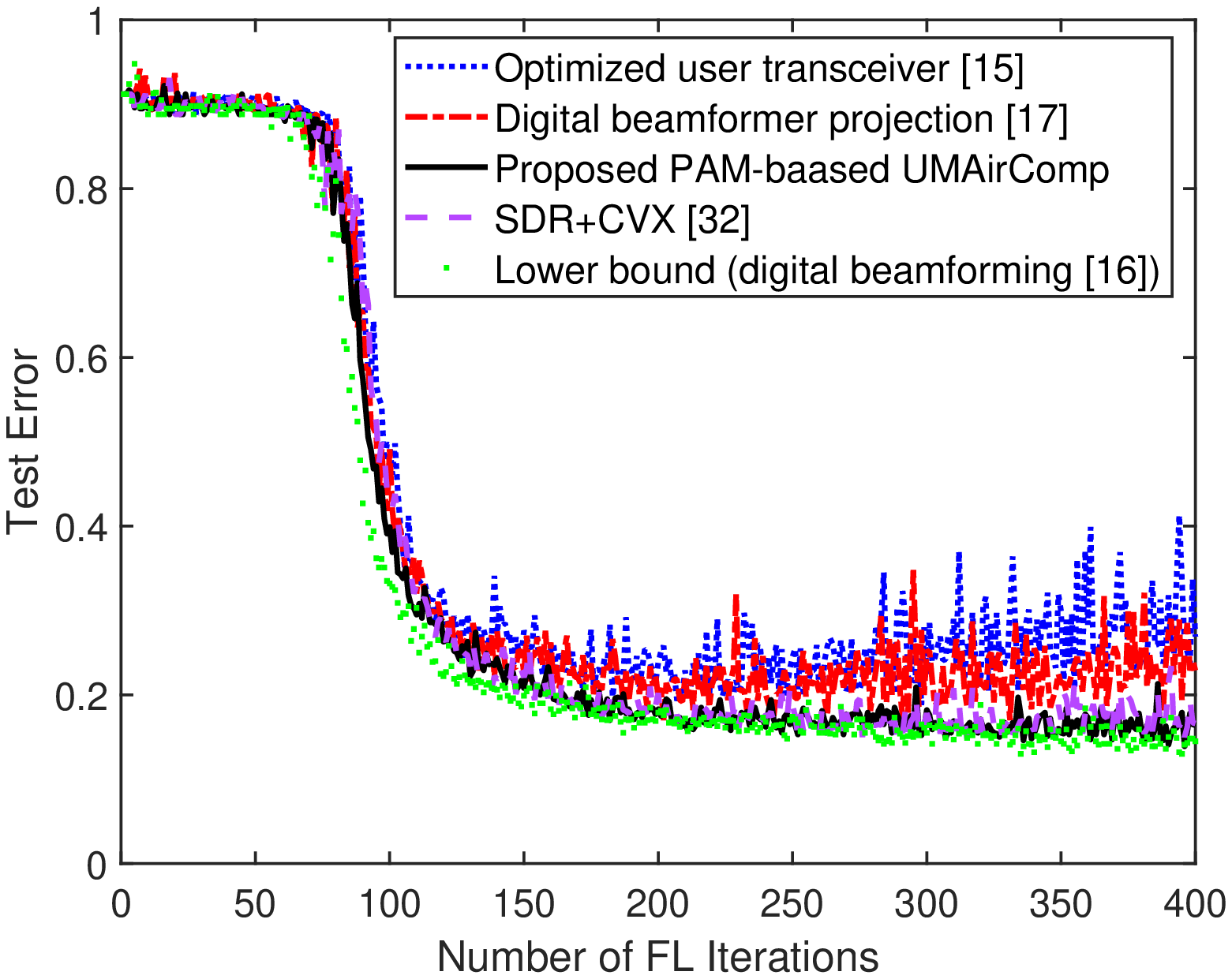}}
  \caption{Comparison between the proposed and benchmark schemes when $N=8$ with $K=10$: a) comparison of normalized MSE, training loss, and test error; b) training loss versus the number of FL iterations; c) worst test error among all users versus the number of FL iterations.}
\end{figure*}

\begin{figure*}[!t]
\centering
\subfigure[]{\includegraphics[height=65mm]{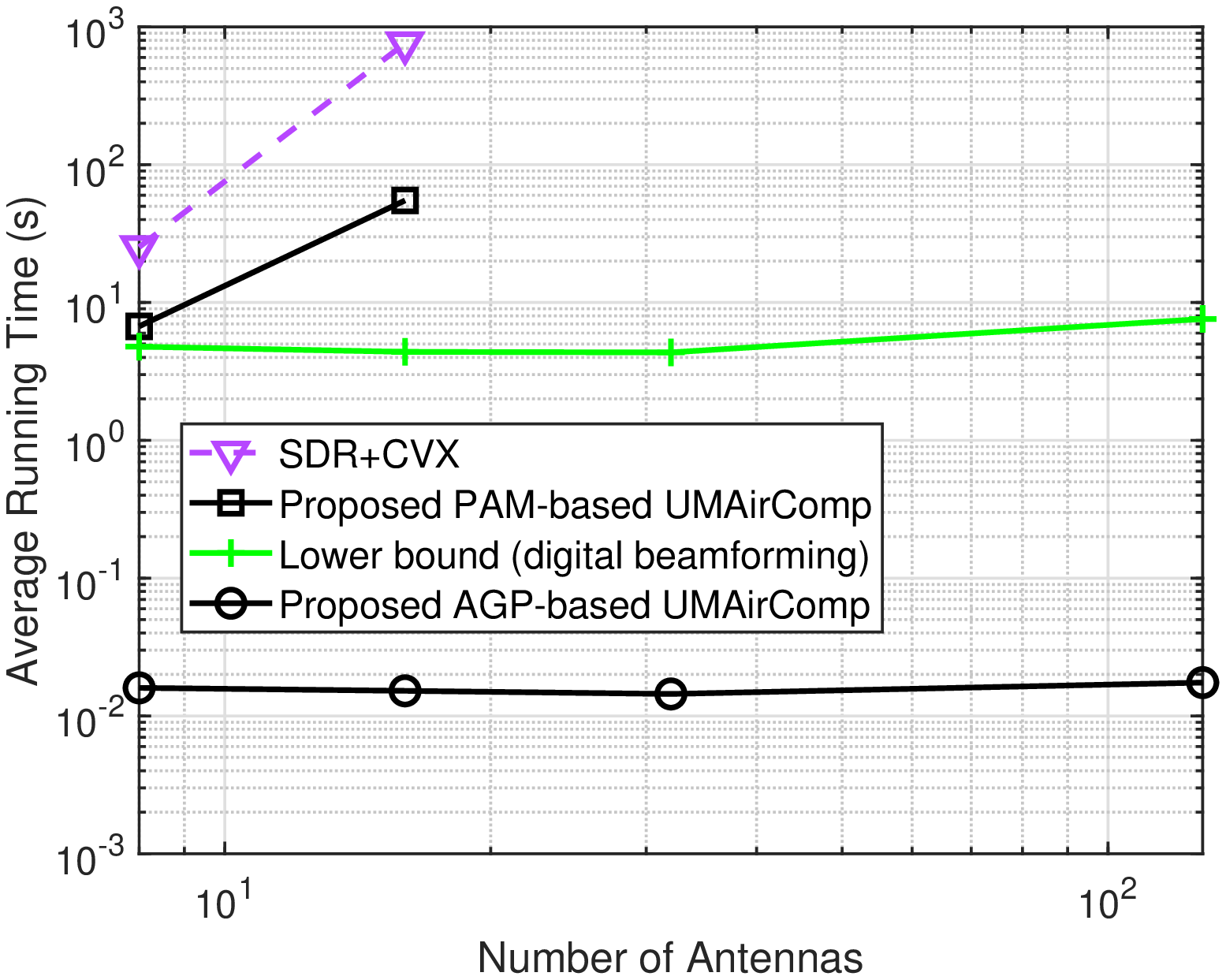}}
\subfigure[]{\includegraphics[height=65mm]{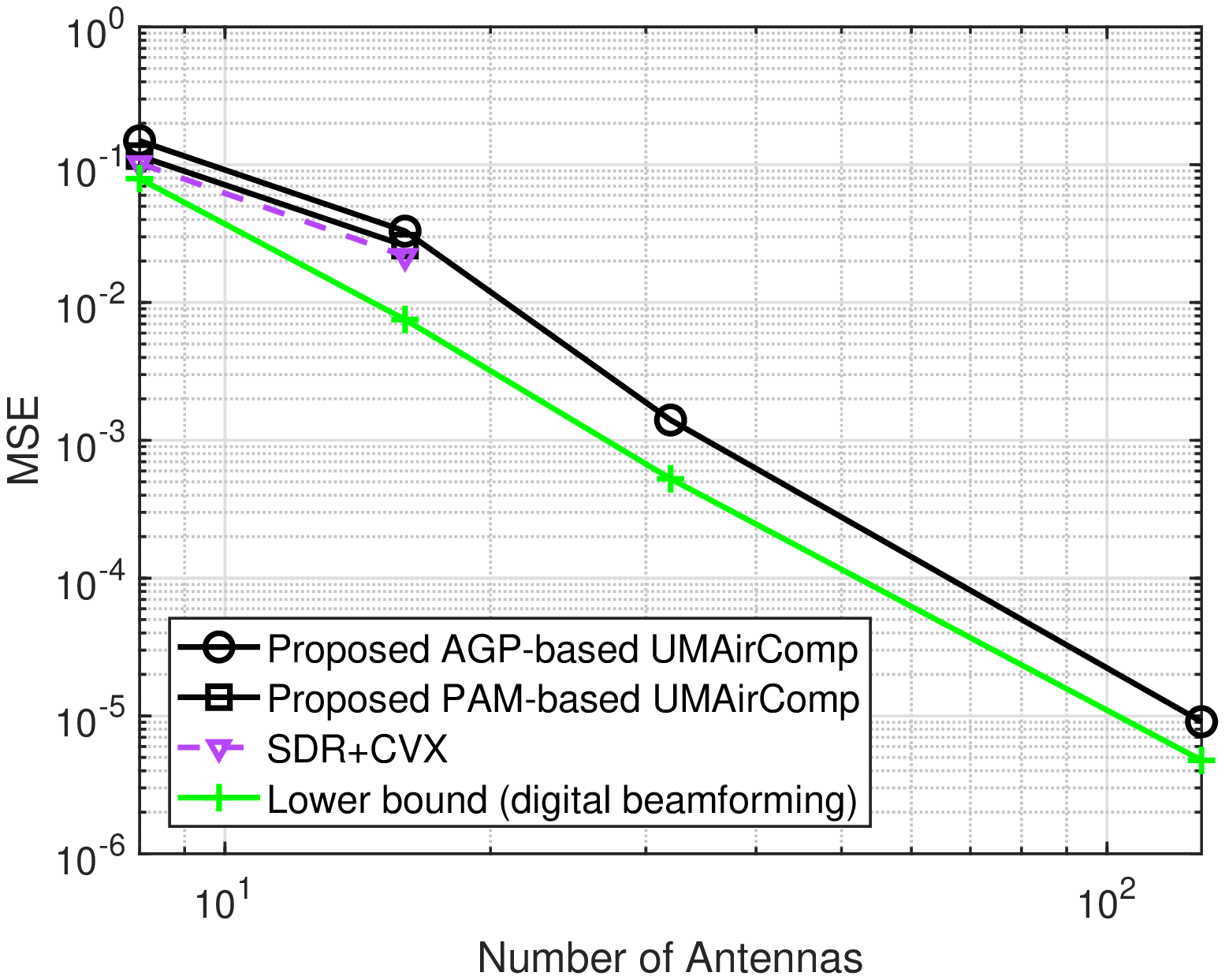}}
\subfigure[]{\includegraphics[height=90mm]{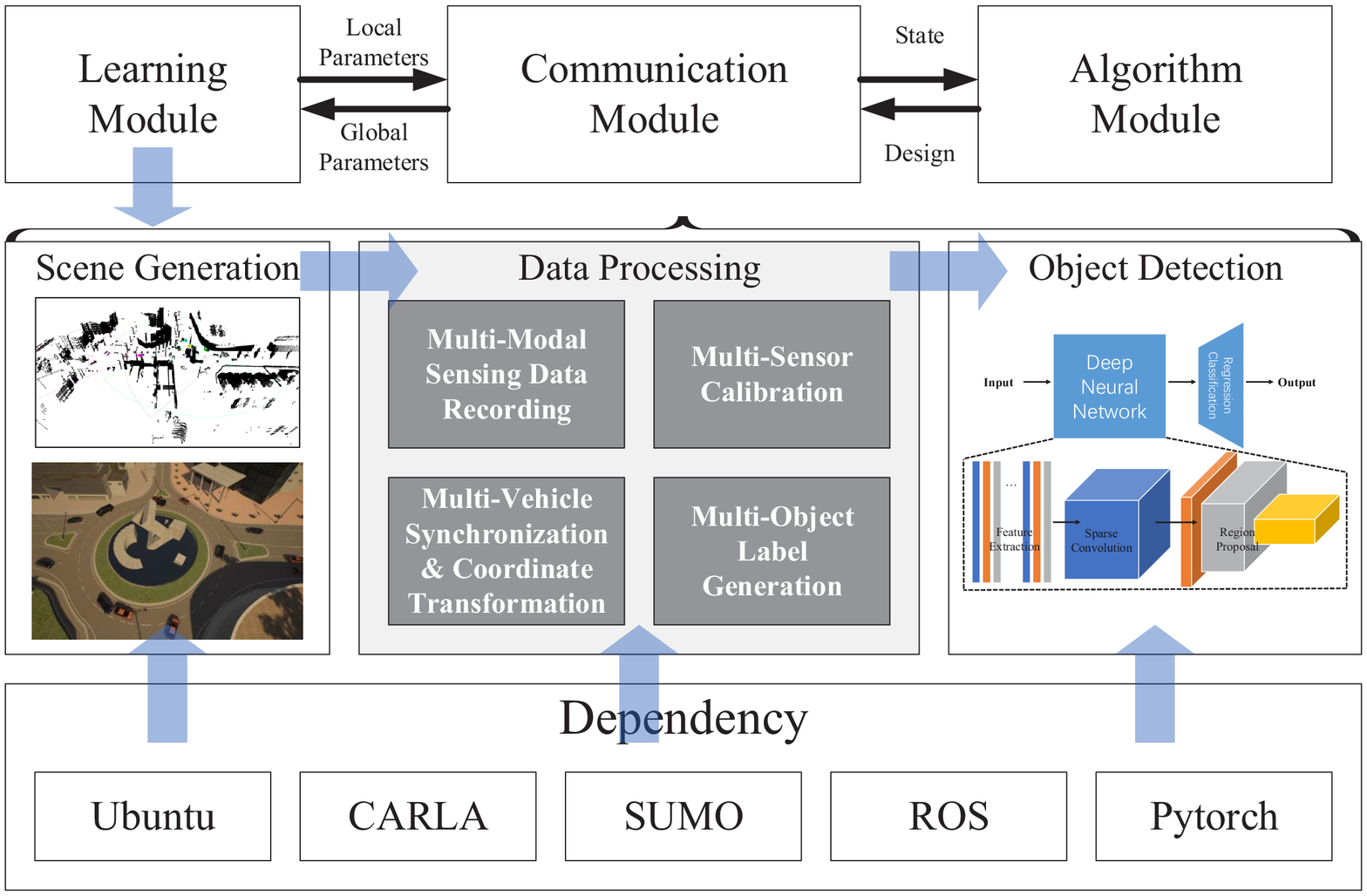}}
\caption{a) Average running time versus the number of antennas when $K=4$; b) normalized MSE versus the number of antennas when $K=4$; c) Architecture of the proposed V2X autonomous driving platform based on CARLA.}
\end{figure*}

To evaluate the solution quality and running time of the proposed AGP-based UMAirComp, we simulate the case of $N=\{8,16,32,128\}$ with $K=4$.
It can be seen from Fig.~5a that the scheme with SDR and CVX is the most time consuming, and it fails in providing a solution within a reasonable amount of time for the case of $N=\{32,128\}$.
The proposed AGP-based UMAirComp requires running time two orders of magnitude smaller than that of all the other schemes.
Note that the lower bound is obtained from the digital beamforming scheme, which tackles a less challenging problem without unit-modulus constraints.
The AGP algorithm achieves MSE performance close to that of the proposed PAM algorithm and the lower bound as shown in Fig.~5b.
This indicates that the adopted approximations in the AGP algorithm lead to small performance loss in practice.

\subsection{UMAirComp for V2X Autonomous Driving}

Next, to verify the robustness of the proposed UMAirComp framework in complex learning tasks, we consider the object detection task for V2X autonomous driving \cite{invs}.
In particular, we propose a V2X autonomous driving simulation platform shown in Fig.~5c, which is a virtual-reality system with close-to-reality environments, interactions, and interfaces.
The platform is developed based on Car Learning to Act (CARLA) \cite{carla}, Simulation of Urban MObility (SUMO), Robot Operating System (ROS), and Pytorch in Ubuntu 18.04 with a GeForce GTX 1080 GPU.
The platform consists of learning, communication, and algorithm modules, where the learning and communication modules interact with each other by exchanging model parameters while the algorithm module controls the communication module by exchanging system/channel state information and optimized designs.

\begin{figure*}[!t]
 \centering
\subfigure[]{\includegraphics[height=50mm]{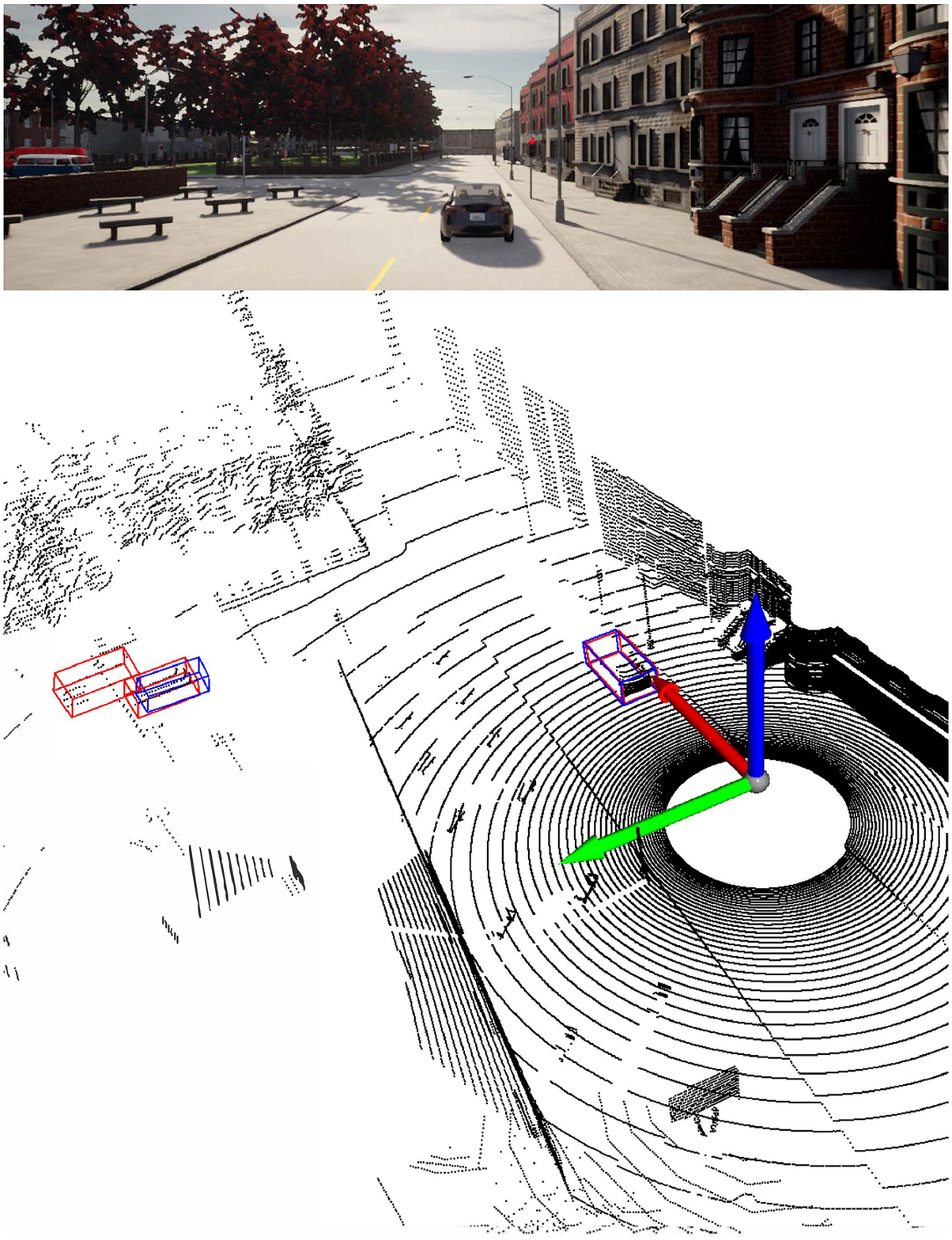}}
\subfigure[]{\includegraphics[height=50mm]{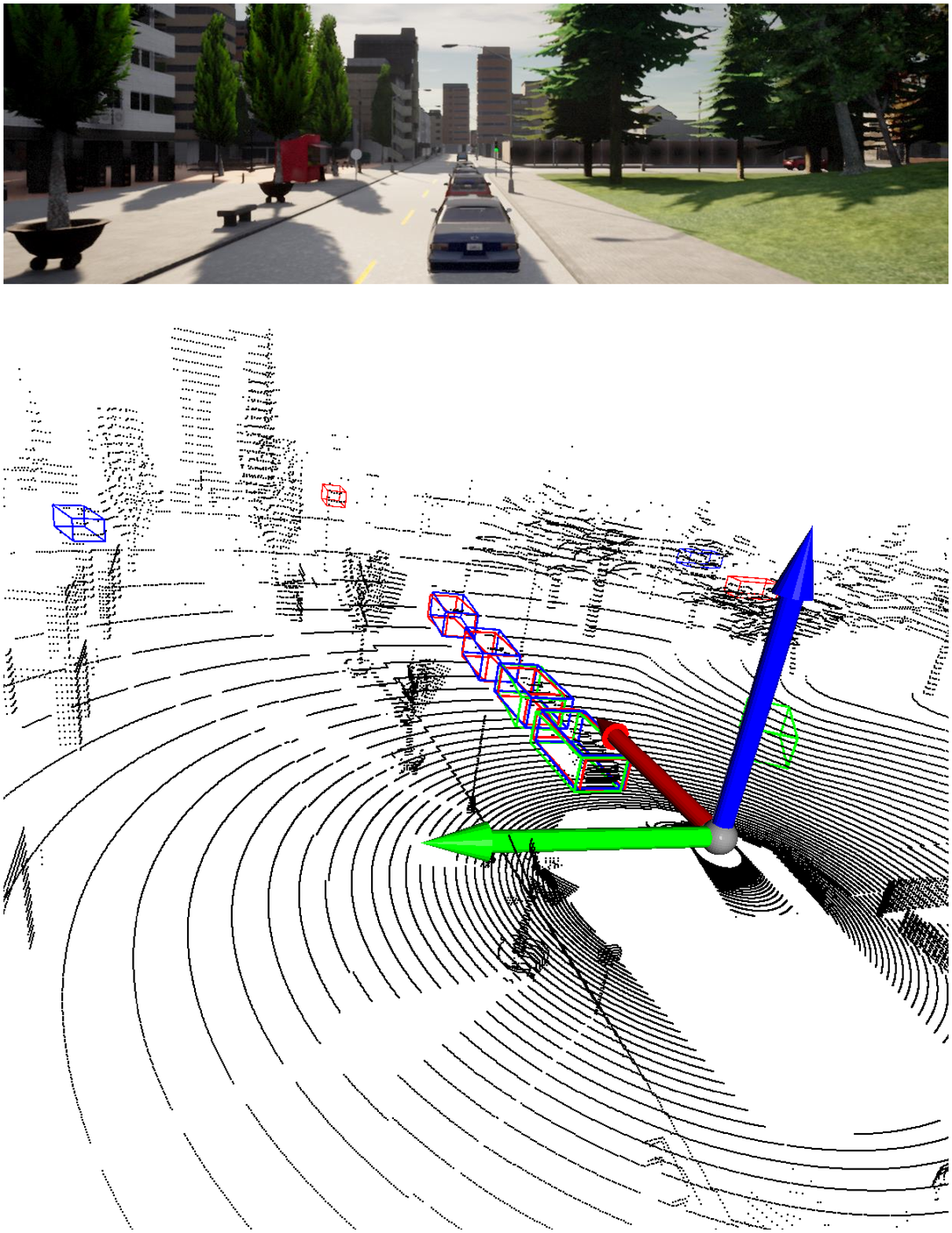}}
\subfigure[]{\includegraphics[height=50mm]{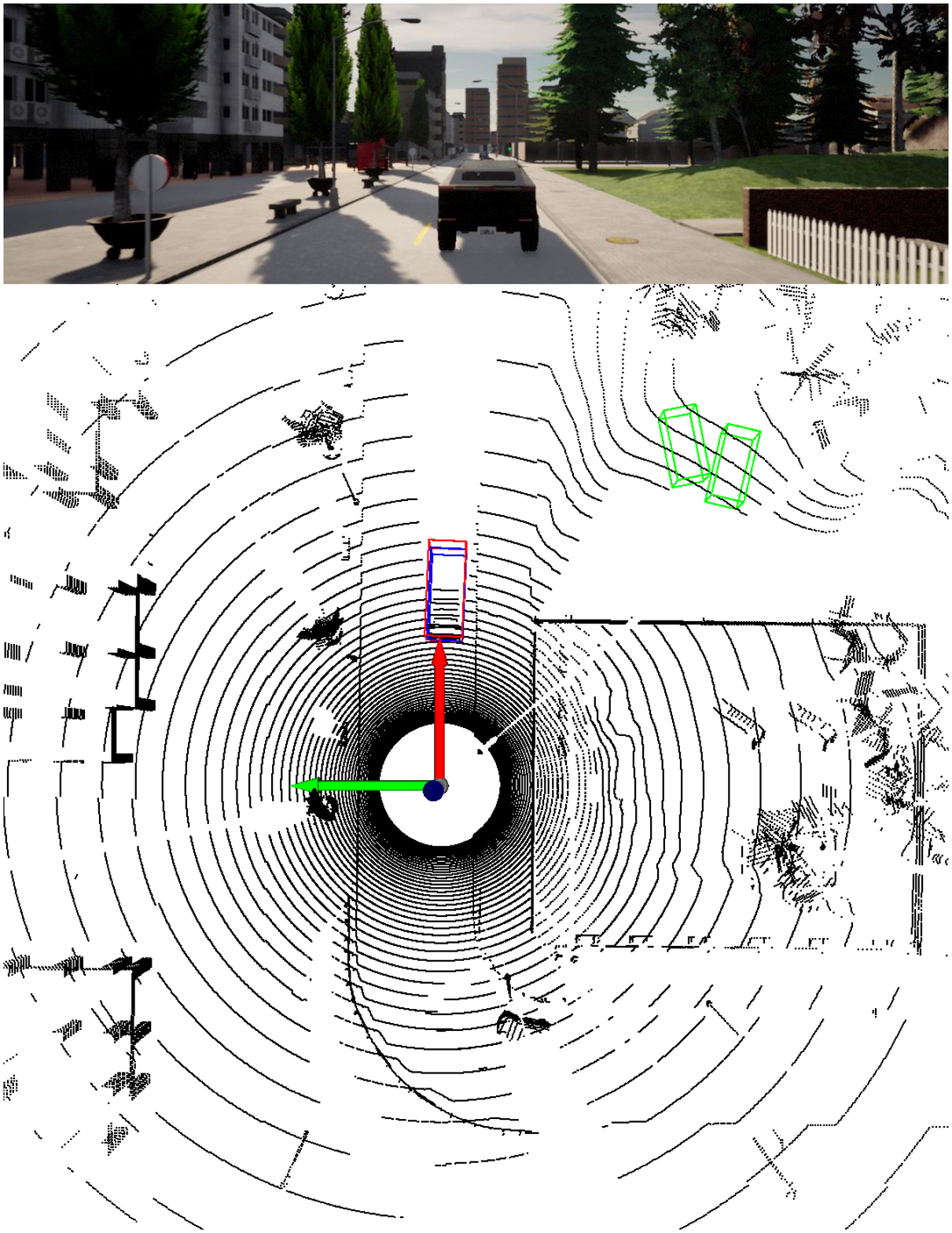}}
\subfigure[]{\includegraphics[height=50mm]{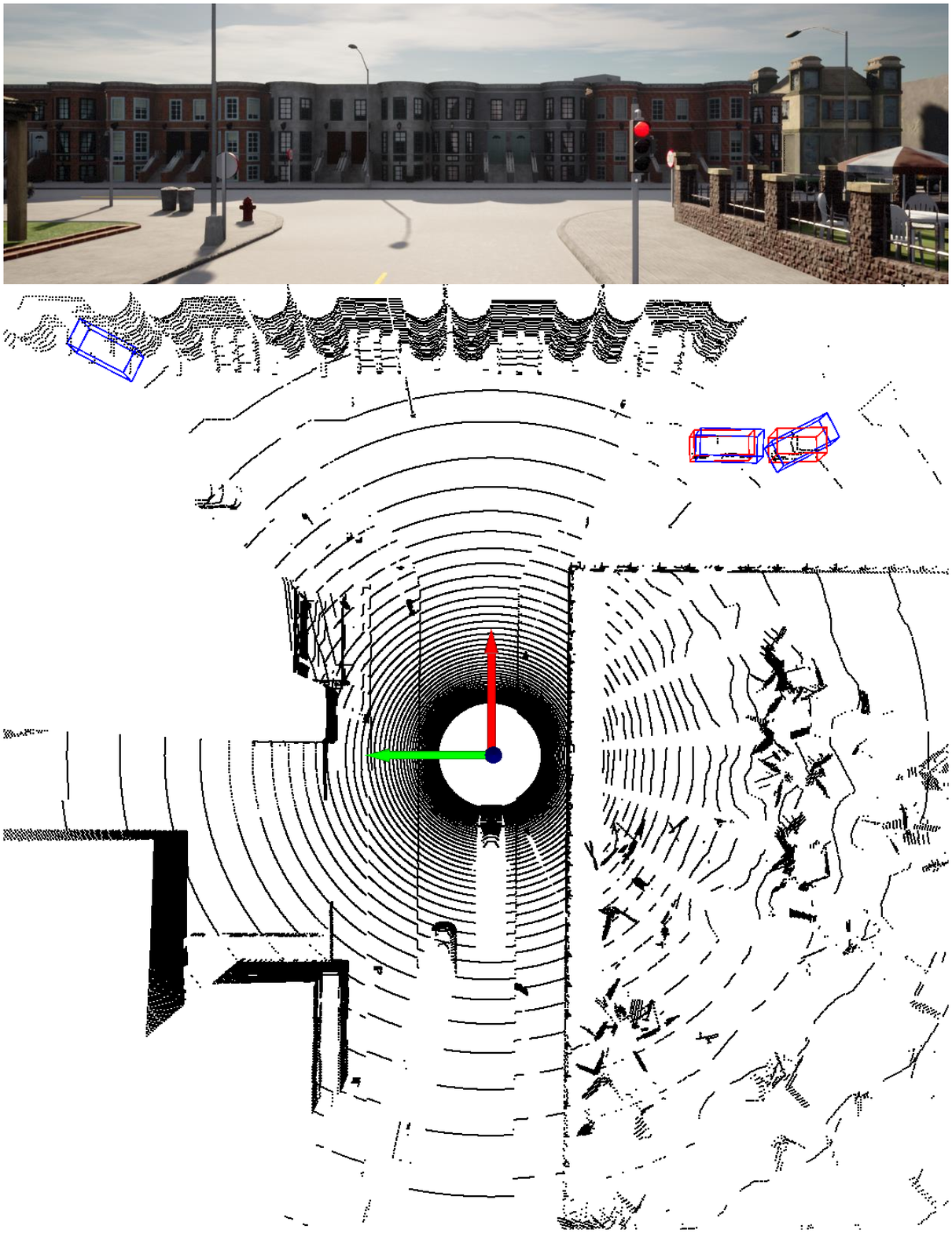}}
  \caption{Detection results when $N=500$ with $K=4$. The red box is the ground truth; the blue box is from the proposed AGP-based UMAirComp scheme; the green box is from the optimized user transceiver scheme (benchmark) scheme.
  a) the benchmark scheme detects nothing while the proposed scheme detects two objects; b) the benchmark scheme only detects two nearby objects while the proposed scheme can detect far-away objects; c) the benchmark scheme generates false positive results while the proposed scheme generates accurate prediction; d) the benchmark scheme cannot detects occlusion objects while the proposed scheme detects all of them.}
\end{figure*}

\begin{figure*}[!t]
 \centering
\subfigure[]{\includegraphics[height=70mm]{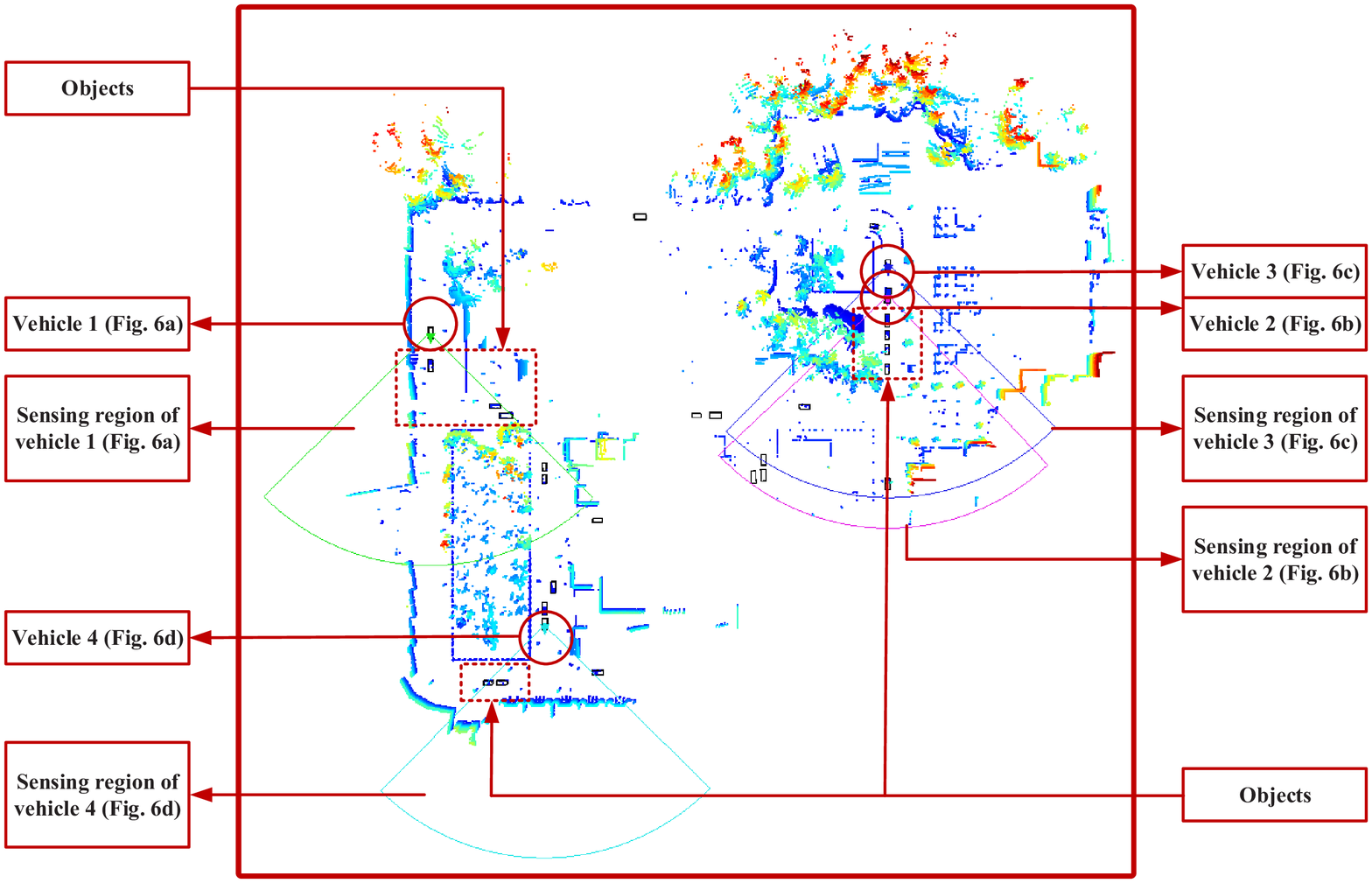}}
\subfigure[]{\includegraphics[height=70mm]{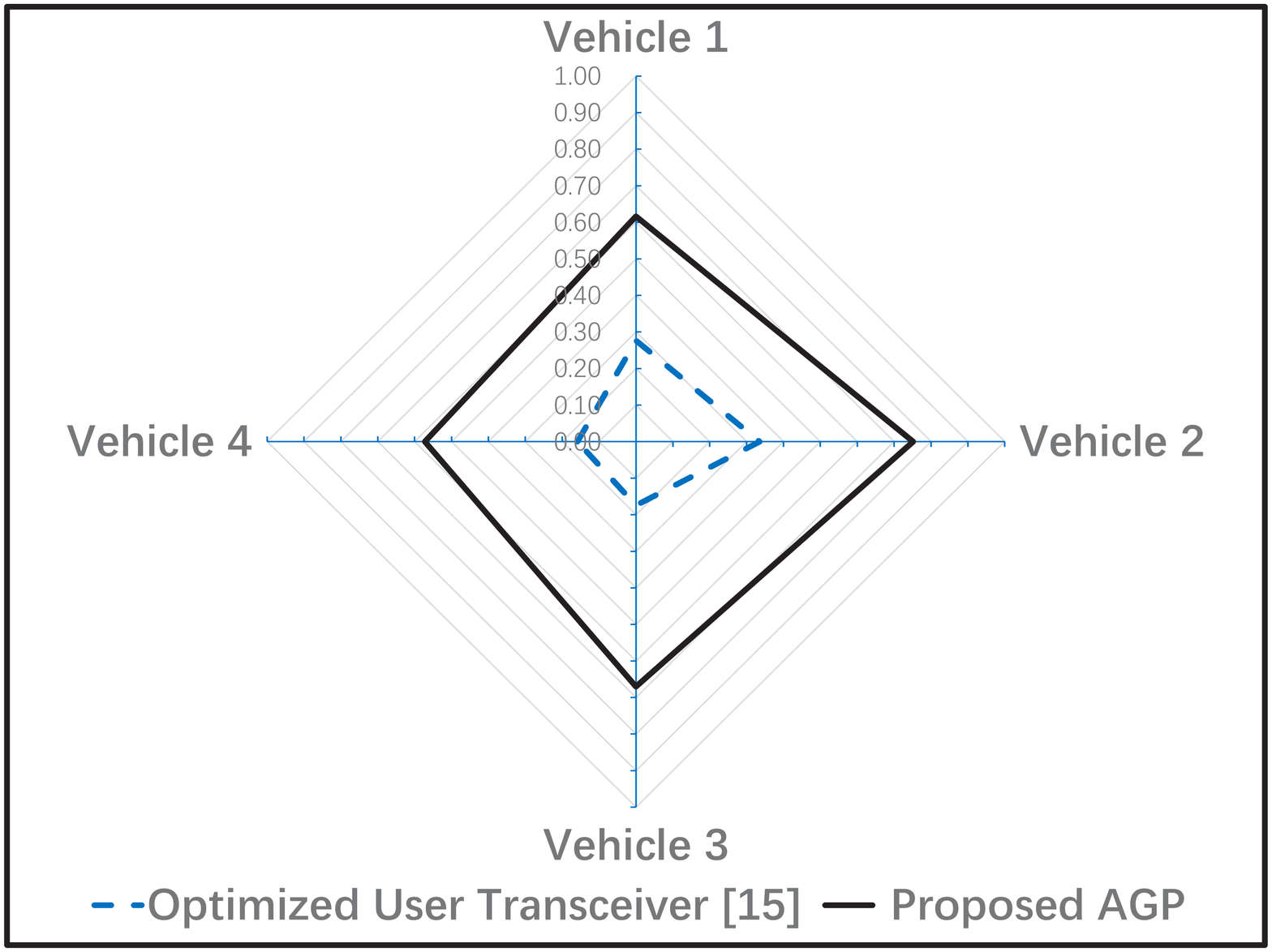}}
  \caption{a) The global bird eye view of the frame in Fig.~7. The green vehicle is in Fig.~6a. The pink vehicle is in Fig.~6b. The blue vehicle is in Fig.~6c. The cyan vehicle is in Fig.~6d.
  The figure demonstrates the heterogeneity of the distributed datasets at different vehicles. b) Comparison between the proposed AGP-based UMAirComp and the optimized user transceiver schemes when $N=500$ with $K=4$.}
\end{figure*}

The learning module is a function cluster consisting of scenario generation, data processing, and object detection.
We use CARLA, which is a client-server platform based on Unreal Engine, to construct a high-fidelity scenario.
The ``Town02'' map is generated with $28$ vehicles, including $4$ autonomous vehicles.
The entire scenario lasts for $90$ seconds and contains $700$ frames at each autonomous vehicle.
The first $200$ frames are used for training at each vehicle.
The latter $500$ frames are used for inference and testing.
On the other hand, data processing involves the following steps.
\begin{itemize}
\item[1)] \textbf{Multi-modal sensory data recording}.
Recording sensory data is implemented based on Python application programming interfaces (APIs) at the CARLA client \cite{carla}.
Each autonomous vehicle (Tesla Model 3) is equipped with a GPS, a camera, and a $64$-line LiDAR at $10$ Hz.
The default LiDAR range is set to $100$ m, and its FoV is $90^{\circ}$ in the front.
Each frame of point cloud is considered as an input data sample $\mathbf{d}_{k,l}^{\rm{in}}$.
All the above information is stored in a sensor file.
\item[2)] \textbf{Multi-sensor calibration}.
Since each vehicle has three coordinate systems (i.e., LiDAR coordinate, camera coordinate, and GPS coordinate), we need to compute the rotation and transition matrices among these coordinate systems \cite{kitti}.
Furthermore, since we have multiple cameras, we need to compute the ration matrices from the reference camera to other cameras.
Lastly, projection matrices are required to associate a point in the 3D space with a point in the camera image.
All the above information is stored in a calibration file.
\item[3)] \textbf{Multi-vehicle synchronization and coordinate transformation}.
All vehicles should be synchronized \cite{tits}.
We adopt the synchronization framework provided by CARLA-SUMO integration to synchronize the sensors on each vehicle.
We exploit the time stamps of the 3D laser scanner as a reference and consider each spin as a frame.
On the other hand, different vehicles have different local coordinate systems.
As such, we use the APIs provided by the CARLA worlds, maps, and actors to transform these local coordinates to global coordinates.
All the above information is stored in a global coordinate file.
\item[4)] \textbf{Multi-object label generation}.
The ground truth labels $\{\mathbf{d}_{k,l}^{\mathrm{out}}\}$ for object detection should satisfy the KITTI format.
Each label $\mathbf{d}_{k,l}^{\mathrm{out}}$ consists of 16 elements \cite{kitti}:
1) category; 2) boundary or not; 3) occlusion or not; 4) observation angle; 5--6) left top pixel coordinates in the camera image; 7--8) right bottom pixel coordinates in the camera image; 9) height of object; 10) width of object; 11) length of object; 12-14) 3D object location in the camera coordinate system; 15) rotation of object; 16) detection confidence.
All the above information is stored in a label file.\footnote{In practice, to obtain the training labels of objects, each vehicle broadcasts its ego position and waits for messages from the nearby infrastructures.
Since the infrastructures are fixed at utility poles and connected to servers via wirelines, they have broader fields of views (FoVs) and deeper neural networks than those of vehicles.
Thus, their outputs are more accurate, which can be transmitted to vehicles via the vehicle-to-infrastructure (V2I) communication and adopted as pseudo labels (i.e., ground-truth labels with noises).}
\end{itemize}

Finally, the sparsely embedded convolutional detection (SECOND) neural network \cite{SECOND} is used for object detection on the processed dataset.
The local model structure can be found in \cite{SECOND}.
The loss function
\begin{align}
\Theta(\mathbf{d}_{k,l},\mathbf{x}_k)=&f_{\mathrm{class}}\left(\mathbf{x}_k,\mathbf{d}_{k,l}\right)+
f_{\mathrm{box}}\left(\mathbf{x}_k,\mathbf{d}_{k,l}\right)
+f_{\mathrm{orientation}}\left(\mathbf{x}_k,\mathbf{d}_{k,l}\right),
\end{align}
where $\mathbf{d}_{k,l}=\left(\mathbf{d}_{k,l}^{\rm{in}},\mathbf{d}_{k,l}^{\rm{out}}\right)$, $f_{\mathrm{class}}$ is the classification loss, $f_{\mathrm{box}}$ is the box regression loss, and $f_{\mathrm{orientation}}$ is the softmax orientation estimation loss.
The SECOND network is trained with a diminishing learning rate, where the initial learning rate is set to $10^{-4}$ and the number of local updates is $E=1$.
The average precision at intersection over union IoU$=0.5$ is used for performance evaluation.

For the communication module, the case of $N=500$ and $K=4$ is simulated.
It is assumed that the sensing datasets are generated and pre-stored at the vehicles before the FL procedure.
The total number of FL iterations is $8$, with independently generated channel in each FL iteration.
The maximum transmit power of each autonomous vehicle is set to $P_0=1\,\mathrm{W}$ (i.e., $30~\mathrm{dBm}$), and the noise powers are set as $-90\,\mathrm{dBm}$.

The detection results and the global bird eye view of a particular frame are shown in Fig.~6 and Fig.~7a, respectively.
The comparison between the proposed scheme and the optimized user transceiver scheme \cite{air3} is provided in Fig.~7b.
From the above results, it can be seen that the test errors of both schemes in Fig.~7b are worse than their counterparts in Fig.~4a.
In contrast, the MSEs of both schemes are remarkably better than their counterparts in Fig.~4a.\footnote{The normalized MSEs of the proposed AGP-based UMAirComp and the optimized user transceiver (benchmark) schemes in Fig.~7b are only $6\times10^{-11}$ and $0.0013$, respectively. The normalized MSEs of the two schemes in Fig.~4a are larger than 0.05.}
This implies that object detection tasks in autonomous driving are more sensitive to model parameter errors.
This is because the trained model parameters for the SECOND network are generally sparser and a slight model error would lead to completely different detections.
Furthermore, as seen from the objective of $\mathcal{P}$, the model error of user $k$ with the UMAirComp framework is dominated by $\gamma\sum_{j=1}^K\Big|r_k\mathbf{g}_k^H\mathbf{F}\mathbf{h}_jt_j-\alpha_j\Big|^2$.
Therefore, the methods to minimize the model errors can be categorized into 1) configuration of wireless channels $\{\mathbf{h}_{k},\mathbf{g}_{k}\}$ and 2) design of analog beamformer $\mathbf{F}$.
For the configuration of wireless channels, due to $\alpha_1=\cdots=\alpha_K$, the ideal wireless channels for the proposed UMAirComp satisfy $\mathbf{g}_1=\cdots=\mathbf{g}_K$ and $\mathbf{h}_1=\cdots=\mathbf{h}_K$.
This implies that the UMAirComp should be executed when vehicles are geographically close to each other such that the magnitudes of $\{\mathbf{h}_{k},\mathbf{g}_{k}\}$ are close.
This is the case in vehicle platooning and vehicle parking scenarios \cite{tits}; otherwise, some emerging techniques (e.g., reconfigurable intelligence surfaces (RIS) \cite{edge3}) should be adopted to smartly alter the wireless environment.
On the other hand, for the design of beamformers, the key is to align the various channels $\{\mathbf{h}_{k},\mathbf{g}_{k}\}$ to a same direction and power for decoding superimposed signals.
This is the case of Fig.~6b, where the proposed method achieves much larger average precisions than the optimized user transceiver scheme for all vehicles.
Note that the running time of AGP-based UMAirComp is only $0.06$ s, which can be further accelerated via GPU.

\section{Conclusion}

This paper proposed the UMAirComp framework to support simultaneous transmission of local model parameters in edge federated learning systems.
Training loss upper bounds of UMAirComp were derived, which reveal that the key to minimize FL training loss is to minimize the maximum MSE among all users.
Two low-complexity large-scale optimization algorithms were proposed to tackle the nonconvex nonsmooth loss bound minimization problem.
The performance and runtime of the UMAirComp framework with the proposed optimization algorithms were verified using the image classification task.
The performance of the proposed framework and algorithms were also verified in a V2X autonomous driving simulation platform and experimental results have shown that the object detection precision with the proposed algorithm is significantly higher than that achieved by benchmark schemes.

\appendices

\section{Proof of Theorem 1}
First, let $\Delta \mathbf{x}^{[i]} =\mathbf{x}_k^{[i+1]}(0)-\left[\mathbf{x}_k^{[i]}(0)-\varepsilon\nabla_{\mathbf{x}}\Lambda\left(\mathbf{x}_k^{[i]}(0)\right)\right]$ and $\|\Delta\mathbf{x}^{[i]}\|_2^2$ can be upper bounded by
\begin{align}
\|\Delta\mathbf{x}^{[i]}\|_2^2
&=\Big\|\mathbf{x}^{[i+1]}_k(0)-\mathbf{x}^{[i]}_k(0)+\varepsilon\nabla_{\mathbf{x}}\Lambda\left(\mathbf{x}^{[i]}_k(0)\right)\Big\|_2^2
\nonumber\\
&=
\Big\|
\mathbf{x}^{[i+1]}_k(0)-\bm{\theta}^{[i]}+
\sum_{j=1}^K\alpha_j\mathbf{x}^{[i]}_j(0)-
\varepsilon\sum_{j=1}^{K}\frac{\alpha_j}{|\mathcal{D}_j|}\sum_{\mathbf{d}_{j,l}\in\mathcal{D}_j}\nabla_{\mathbf{x}}\Theta\left(\mathbf{d}_{j,l},\mathbf{x}^{[i]}_j(0)\right)
\nonumber\\
&\quad{}{}
-\mathbf{x}_k^{[i]}(0)+\varepsilon\nabla_{\mathbf{x}}\Lambda(\mathbf{x}_k^{[i]}(0))
\Big\|_2^2
\nonumber\\
&\leq
\Big\|
\mathbf{x}^{[i+1]}_k(0)-\bm{\theta}^{[i]}\Big\|^2_2+
\sum_{j=1}^K\alpha_j\Big\|\mathbf{x}^{[i]}_j(0)-\mathbf{x}^{[i]}_k(0)\Big\|_2^2
\nonumber\\
&\quad{}{}
+
\varepsilon^{2}
\sum_{j=1}^K
\Bigg\|
\frac{\alpha_j}{|\mathcal{D}_j|}\left(
\sum_{\mathbf{d}_{j,l}\in\mathcal{D}_j}\nabla_{\mathbf{x}}\Theta\left(\mathbf{d}_{j,l},\mathbf{x}^{[i]}_j(0)\right)
-
\sum_{\mathbf{d}_{j,l}\in\mathcal{D}_j}\nabla_{\mathbf{x}}\Theta\left(\mathbf{d}_{j,l},\mathbf{x}^{[i]}_k(0)\right)\right)\Bigg\|_2^2,
\label{A1}
\end{align}
where the second equality is obtained from \eqref{sgd} with $E=1$ and \eqref{global}, and the inequality is obtained due to $\|\mathbf{a}_1+\mathbf{a}_2 \|_2^2\leq\|\mathbf{a}_1\|_2^2+\|\mathbf{a}_2 \|_2^2$.

On the other hand, according to \eqref{mse}, we have
\begin{align}
&\mathbb{E}\left[\Big\|\mathbf{x}^{[i+1]}_k(0)-\bm{\theta}^{[i]}\Big\|_2^2\right]=\mathbb{MSE}^{[i]}_k,
\nonumber\\
&
\mathbb{E}\left[\Big\|\mathbf{x}^{[i]}_k(0)-\mathbf{x}^{[i]}_j(0)\Big\|_2^2\right]=
\mathbb{E}\left[\Big\|\mathbf{x}^{[i]}_k(0)-\bm{\theta}^{[i-1]}+\bm{\theta}^{[i-1]}-\mathbf{x}^{[i]}_j(0)\Big\|_2^2\right]
\leq2\mathbb{MSE}^{[i]}_k.
 \label{A2}
\end{align}
Moreover, according to Assumption 1, we  have
\begin{align}
&\mathbb{E}\left[
\Bigg\|\frac{1}{|\mathcal{D}_j|}\left(
\sum_{\mathbf{d}_{j,l}\in\mathcal{D}_j}\nabla_{\mathbf{x}}\Theta\left(\mathbf{d}_{j,l},\mathbf{x}^{[i]}_j(0)\right)
-
\sum_{\mathbf{d}_{j,l}\in\mathcal{D}_j}\nabla_{\mathbf{x}}\Theta\left(\mathbf{d}_{j,l},\mathbf{x}^{[i]}_k(0)\right)\right)\Bigg\|_2^2
\right]
\nonumber\\
&\leq \mathbb{E}\left[L^2\|\mathbf{x}^{[i]}_k(0)-\mathbf{x}^{[i]}_j(0)\|_2^2\right]
\leq 2L^2\,\mathbb{MSE}^{[i]}_k.
\label{A3}
\end{align}
Putting \eqref{A2} and \eqref{A3} into \eqref{A1}, and according to the expression of $\varepsilon$ yields
\begin{align}
&\mathbb{E}\left[\|\Delta \mathbf{x}^{[i]}\|^2_2\right]\leq
\left(3+2\sum_{j=1}^K\alpha_j^2\right)
\,\mathop{\mathrm{max}}_{k=1,\cdots,K}\mathbb{MSE}^{[i]}_k. \label{deltax}
\end{align}

Due to $\frac{1}{|\mathcal{D}_k|}\sum_{\mathbf{d}_{k,l}\in\mathcal{D}_k}\nabla^2_{\mathbf{x}}\Theta(\mathbf{d}_{k,l}, \mathbf{x})\preceq L\mathbf{I}$,
we have $\nabla^2_{\mathbf{x}}\Lambda(\mathbf{x})\preceq 1/L\mathbf{I}$.
Based on $\mu\mathbf{I}\preceq\nabla^2_{\mathbf{x}}\Lambda(\mathbf{x})\preceq 1/L\mathbf{I}$, the following equations hold \cite{yurii}:
\begin{subequations}
\begin{align}
\Lambda(\mathbf{x}') &\leq \Lambda(\mathbf{x}) + (\mathbf{x}'-\mathbf{x})^T\nabla\Lambda(\mathbf{x})+\frac{L}{2}\|\mathbf{x}'-\mathbf{x}\|_2^2,
\label{Lambda leq}\\
\Lambda(\mathbf{x}') &\geq \Lambda(\mathbf{x}) + (\mathbf{x}'-\mathbf{x})^T\nabla\Lambda(\mathbf{x})+\frac{\mu}{2}\|\mathbf{x}'-\mathbf{x}\|_2^2.
\label{Lambda geq}
\end{align}
\end{subequations}
Putting $\mathbf{x}'=\mathbf{x}_k^{[i+1]}(0)=\mathbf{x}_k^{[i]}(0)-\varepsilon\nabla_{\mathbf{x}}\Lambda\left(\mathbf{x}_k^{[i]}(0)\right)+\Delta\mathbf{x}^{[i]}$ and $\mathbf{x}=\mathbf{x}_k^{[i]}(0)$ into
\eqref{Lambda leq}, we have
\begin{align}
&\Lambda\left(\mathbf{x}_k^{[i+1]}(0)\right)\leq
\Lambda\left(\mathbf{x}_k^{[i]}(0)\right)-\frac{1}{2L}\Big\|\nabla_{\mathbf{x}}\Lambda\left(\mathbf{x}_k^{[i]}(0)\right)\Big\|_2^2+\frac{L}{2}\|\Delta\mathbf{x}^{[i]}\|_2^2.
\label{Lambda condition1}
\end{align}
On the other hand, the right hand side of \eqref{Lambda geq} is minimized at
$\mathbf{x}'=\mathbf{x}-\mu^{-1}\nabla_{\mathbf{x}}\Lambda(\mathbf{x})$.
Putting this expression and $\mathbf{x}=\mathbf{x}_k^{[i]}(0)$ into \eqref{Lambda geq} gives
\begin{align}
&\|\nabla_{\mathbf{x}}\Lambda(\mathbf{x}_k^{[i]}(0))\|_2^2
\geq
2\mu\left[\Lambda(\mathbf{x}_k^{[i]}(0))-\Lambda(\bm{\theta}^{*}) \right].
\label{Lambda condition2}
\end{align}
Combining \eqref{Lambda condition1} and \eqref{Lambda condition2} gives
\begin{align}
&\Lambda(\mathbf{x}_k^{[i+1]}(0))-\Lambda(\mathbf{x}_k^{[i]}(0))
\leq
\left(1-\frac{\mu}{L}\right)
\left[\Lambda(\mathbf{x}_k^{[i]}(0))-\Lambda(\bm{\theta}^{*})\right]
+\frac{L}{2}\|\Delta\mathbf{x}^{[i]}\|_2^2.
\end{align}
Applying this recursively leads to
\begin{align}
\Lambda(\mathbf{x}_k^{[i+1]}(0))-\Lambda(\bm{\theta}^{*})
&\leq
\left(1-\frac{\mu}{L}\right)^{i+1}
\left[\Lambda(\mathbf{x}_k^{[0]}(0))-\Lambda(\bm{\theta}^{*})\right]
\nonumber\\
&\quad{}{}
+\sum_{i'=0}^{i}
\frac{L}{2}
\left(1-\frac{\mu}{L}\right)^{i-i'}
\|\Delta\mathbf{x}^{[i']}\|_2^2. \label{A 53}
\end{align}
Taking expectation on both sides and applying \eqref{deltax}, \eqref{A 53} becomes
\begin{align}
\mathbb{E}\left[\Lambda(\mathbf{x}_k^{[i+1]}(0))-\Lambda(\bm{\theta}^*)\right]
&\leq
\left(1-\frac{\mu}{L}\right)^{i+1}
\left[\Lambda(\mathbf{x}_k^{[0]}(0))-\Lambda(\bm{\theta}^{*})\right]
+
\sum_{i'=0}^{i}A^{[i']}\,\mathrm{max}_k\mathbb{MSE}^{[i']}_k.
\end{align}
Finally, setting, $i=R-1$, taking the limit $R\rightarrow+\infty$ and using $\left(1-\mu/L\right)^{R}\rightarrow 0$, the proof is completed.

\section{Proof of Theorem 2}

Let $\mathbf{a}_k^{[t]}$ with $t=0,\cdots,RE-1$ be the parameter vector at user $k$ such that $\mathbf{a}_k^{[t]}=\mathbf{x}_k^{[i]}(\tau)$ for $t=iE+\tau$.
Define the following virtual sequences $\{\widetilde{\mathbf{a}}^{[t]}\}$ as
\begin{align}
&\widetilde{\mathbf{a}}^{[t]}=\sum_{k=1}^K\alpha_k\mathbf{a}_k^{[t]}+\sum_{k=1}^K\alpha_k\bm{\sigma}^{[t]}_k\mathbb{I}_{t\,\mathrm{mod}\,E=0},
\label{E3}
\end{align}
where $\bm{\sigma}^{[t]}_k$ denotes the model distortion at the $k$-th user and the $t$-th iteration, which is upper bounded by $\mathbb{E}[\|\bm{\sigma}^{[t]}_k\|_2^2]\leq\mathop{\mathrm{max}}_{\forall i,k}\mathbb{MSE}_k^{[i]}$.
Based on \eqref{E3} and equations (2) and (10) in Section II of the revised manuscript, we have
\begin{align}
\widetilde{\mathbf{a}}^{[t+1]}&=\widetilde{\mathbf{a}}^{[t]}-
\frac{\varepsilon}{\sum_{k=1}^{K}|\mathcal{D}_k|}\sum_{k=1}^K\sum_{\mathbf{d}_{k,l}\in\mathcal{D}_k}\nabla_{\mathbf{x}_k}\Theta(\mathbf{d}_{k,l}, \mathbf{x}_k)|_{\mathbf{x}_k=\mathbf{a}_k^{[t]}}
\nonumber\\
&\quad{}{}
+\sum_{k=1}^K\alpha_k\left(\bm{\sigma}^{[t+1]}_k\mathbb{I}_{t+1\,\mathrm{mod}\,E=0}-\bm{\sigma}^{[t]}_k\mathbb{I}_{t\,\mathrm{mod}\,E=0}\right).
\end{align}
Since $\varepsilon=\frac{2}{\mu(\nu+iE+\tau)}\leq \frac{1}{4L}$ and according to \cite[Lemma 1]{iclr}, we have
\begin{align}
\mathbb{E}[\|\widetilde{\mathbf{a}}^{[t+1]}-\bm{\theta}^*\|_2^2 ]
&\leq(1-\mu \varepsilon)\mathbb{E}[\|\widetilde{\mathbf{a}}^{[t]}-\bm{\theta}^*\|_2^2 ]
+\sum_{j=1}^K\alpha_j
\mathbb{E}[\|\bm{\sigma}^{[t+1]}_j\mathbb{I}_{t+1\,\mathrm{mod}\,E=0}-\bm{\sigma}^{[t]}_j\mathbb{I}_{t\,\mathrm{mod}\,E=0}\|_2^2]
\nonumber\\
&\quad{}
+\frac{2}{K}\sum_{k=1}^K\mathbb{E}[\|\widetilde{\mathbf{a}}^{[t]}-\mathbf{a}_k^{[t]}\|_2^2 ]
+6L\Gamma\varepsilon^2.
\label{E5}
\end{align}
To further bound the right hand side of \eqref{E5}, we notice that for any $(k,t)$,
\begin{align}
\mathbb{E}[\|\bm{\sigma}^{[t+1]}_k\mathbb{I}_{t+1\,\mathrm{mod}\,E=0}-\bm{\sigma}^{[t]}_k\mathbb{I}_{t\,\mathrm{mod}\,E=0}\|_2^2]\leq\mathbb{E}[\|\bm{\sigma}^{[t]}_k\|^2]\leq\mathop{\mathrm{max}}_{\forall i,k}\mathbb{MSE}_k^{[i]}.
\tag{E6}
\end{align}
Moreover, based on \cite[Lemma A.3]{air1}, the term $\mathbb{E}[\|\widetilde{\mathbf{a}}^{[t]}-\mathbf{a}_k^{[t]}\|_2^2 ]$ is upper bounded as
\begin{align}
&\mathbb{E}[\|\widetilde{\mathbf{a}}^{[t]}-\mathbf{a}_k^{[t]}\|_2^2 ]\leq 4\varepsilon^2G^2E^2.
\end{align}
Consequently, the following recursive error bound is obtained
\begin{align}
&\mathbb{E}[\|\widetilde{\mathbf{a}}^{[t+1]}-\bm{\theta}^*\|_2^2 ]\leq(1-\mu \varepsilon)\mathbb{E}[\|\widetilde{\mathbf{a}}_k^{[t]}-\bm{\theta}^*\|_2^2 ]
+\varepsilon^2C.
\end{align}
where $C$ is defined in Theorem 2.
It can be seen that the above inequality is the same as equation (B.1) in \cite[Lemma A.3]{air1}.
Following the induction procedure in Appendix B of \cite{air1}, the non-recursive error bound is obtained
\begin{align}
&\mathbb{E}[\|\widetilde{\mathbf{a}}^{[t]}-\bm{\theta}^*\|_2^2 ]\leq
\frac{\mathrm{max}(4C/\mu^2,\nu\|\bm{\theta}^{[0]}-\bm{\theta}^*\|_2^2)}{t+\nu}.
\tag{E9}
\end{align}
Multiplying $2L$ on both sides of and using the Lipschitz condition of $\Lambda$ yield
\begin{align}
\mathbb{E}\left[\Lambda(\widetilde{\mathbf{a}}^{[t]})\right]-\Lambda(\bm{\theta}^*)
\leq
2L\mathbb{E}[\|\widetilde{\mathbf{a}}^{[t]}-\bm{\theta}^*\|_2^2 ]
\leq
\frac{2L\,\mathrm{max}(4C,\mu^2\nu\|\bm{\theta}^{[0]}-\bm{\theta}^* \|_2^2)}{\mu^2(t+\nu)},
\label{E10}
\end{align}
Substituting $\widetilde{\mathbf{a}}_k^{[t]}=\mathbf{x}_k^{[R]}(0)$ and $t=ER$ into \eqref{E10}, the proof is completed.

\section{Proof of Lemma 1}

Define a surrogate function of $h(\mathbf{v})$ as
\begin{align}
g(\mathbf{v},\mathbf{v}')=
-\mathop{\sum}_{k=1}^K\frac{b_k|\mathbf{g}_k^H\mathbf{v}|^2}{\sigma_k^2}
+\mathop{\sum}_{k=1}^K\frac{b_k\left(\mathbf{v}-\mathbf{v}'\right)^H\mathbf{g}_k\mathbf{g}_k^H\left(\mathbf{v}-\mathbf{v}'\right)}{\sigma_k^2}. \label{prox}
\end{align}
Since $g(\mathbf{v},\mathbf{v}')\geq h(\mathbf{v})$, $g(\mathbf{v}',\mathbf{v}')=h(\mathbf{v}')$, and $\nabla g(\mathbf{v}',\mathbf{v}')\geq \nabla h(\mathbf{v}')$,
with any feasible $\mathbf{v}^{(0)}$, every limit point of the sequence $(\mathbf{v}^{(0)},\mathbf{v}^{(1)},\cdots)$ generated by the following iteration
\begin{align}
&\mathbf{v}^{(n+1)}=\mathop{\mathrm{arg~min}}_{\|\mathbf{v}\|_2^2\leq \beta}~\mathop{\mathrm{max}}_{\mathbf{b}\in\Delta}\,g(\mathbf{v},\mathbf{v}^{(n)})
\label{mm}
\end{align}
is the KKT solution to \eqref{P2'}.
Therefore, to prove the Lemma, it remains to show that $U(\mathbf{v}')$ is the optimal solution to \eqref{mm}.
Specifically, applying the quasi-concave-convex property of \eqref{mm} and the general minimax theorem \cite{minimax}, we have
\begin{align}
&\mathop{\mathrm{min}}_{\|\mathbf{v}\|_2^2\leq \beta}~\mathop{\mathrm{max}}_{\mathbf{b}\in\Delta}\,g(\mathbf{v},\mathbf{v}')
=
\mathop{\mathrm{max}}_{\mathbf{b}\in\Delta}~\mathop{\mathrm{min}}_{\|\mathbf{v}\|_2^2\leq \beta}\,g(\mathbf{v},\mathbf{v}'). \label{minimax}
\end{align}
Via the Lagrange multiplier method, it can be derived that
\begin{align}
&
\mathop{\mathrm{arg~min}}_{||\mathbf{v}||_2^2\leq \beta}\,g(\mathbf{v},\mathbf{v}')
=\frac{\sqrt{\beta}\mathbf{C}(\mathbf{v}')\mathbf{b}}{\big\|\mathbf{C}(\mathbf{v}')\mathbf{b}\big\|_2}.
\end{align}
Putting the above result into $g(\mathbf{v},\mathbf{v}')$, we have
\begin{align}
&
g\left(\frac{\sqrt{\beta}\mathbf{C}(\mathbf{v}')\mathbf{b}}{\big|\big|\mathbf{C}(\mathbf{v}')\mathbf{b}\big|\big|_2},\mathbf{v}'\right)
=-\Phi\left(\mathbf{v}', \mathbf{b}\right).
\end{align}
Therefore, the optimal solution of $\mathbf{v}$ to \eqref{minimax} (thus \eqref{mm}) is
$U(\mathbf{v}')$ and the proof is completed.

\section{Proof of Lemma 2}

To prove this theorem, we first introduce the following lemma.
\begin{lemma}
(\cite[Lemma 1.2.2]{yurii}) If $h(\mathbf{x})$ is convex and twice differentiable, then $h(\mathbf{x})$ is Lipschitz smooth with constant $L$ if and only if $\nabla^2 h(\mathbf{x})\preceq L\,\mathbf{I}$.
\end{lemma}

Based on Lemma 4 and since $\Xi(\mathbf{b})$ is convex and twice differentiable, it suffices to show $\nabla^2\Xi(\mathbf{b})\preceq L_{\Xi}(\phi)\,\mathbf{1}_K$.
In particular, according to \eqref{Xi}, the Hessian matrix of $\Xi(\mathbf{b})$ is
\begin{align}
\nabla^2\Xi(\mathbf{b})=&
\frac{2\sqrt{\beta}\mathrm{Re}\left(\mathbf{C}^H\mathbf{C}\right)
}{\sqrt{\phi^2+||\mathbf{C}\mathbf{b}||_2^2}}
-\frac{2\sqrt{\beta}}{(\sqrt{\phi^2+||\mathbf{C}\mathbf{b}||_2^2})^3}
\times \mathrm{Re}\left(\mathbf{C}^H\mathbf{C}\mathbf{b}\right)
[\mathrm{Re}\left(\mathbf{C}^H\mathbf{C}\mathbf{b}\right)]^H.
\nonumber
\end{align}
Due to $ \mathrm{Re}\left(\mathbf{C}^H\mathbf{C}\mathbf{b}\right)
[\mathrm{Re}\left(\mathbf{C}^H\mathbf{C}\mathbf{b}\right)]^H\succeq \mathbf{0}$, we can drop the last term to bound $\nabla^2\Xi(\mathbf{b})$ from above, which leads to
\begin{align}
\nabla^2\Xi(\mathbf{b})&\preceq
\frac{2\sqrt{\beta}}{\sqrt{\phi^2+||\mathbf{C}\mathbf{b}||_2^2}}\times \mathrm{Re}\left(\mathbf{C}^H\mathbf{C}\mathbf{b}\right)
\preceq
\frac{2\sqrt{\beta}\cdot\lambda_{\mathrm{max}}\left[\mathrm{Re}\left(\mathbf{C}^H\mathbf{C}\right)\right]}{\sqrt{\phi^2+||\mathbf{C}\mathbf{b}||_2^2}}\mathbf{I}_K,  \label{D1}
\end{align}
where the second inequality follows from
$\mathrm{Re}\left(\mathbf{C}^H\mathbf{C}\right)\preceq \lambda_{\mathrm{max}}\left[\mathrm{Re}\left(\mathbf{C}^H\mathbf{C}\right)\right]\mathbf{I}_K$.

Now, the only quantity in \eqref{D1} that is dependent on $\mathbf{b}$ is $||\mathbf{C}\mathbf{b}||_2^2$.
To get rid of such dependence, $||\mathbf{C}\mathbf{b}||_2^2$ is lower bounded by
\begin{align}
||\mathbf{C}\mathbf{b}||_2^2
&=\mathbf{b}^T\mathbf{C}^H\mathbf{C} \mathbf{b}
\geq
\lambda_{\mathrm{min}}\left(\mathbf{C}^H\mathbf{C}\right)\times||\mathbf{b}||_2^2. \label{D2}
\end{align}
Finally, using $||\mathbf{b}||_2^2\geq 1/K$ due to $\mathbf{b}\in\Delta$ and Cauchy-Schwarz inequality further leads to
\begin{align}\label{D3}
||\mathbf{C}\mathbf{b}||_2^2
\geq
\lambda_{\mathrm{min}}\left(\mathbf{C}^H\mathbf{C}\right)/K.
\end{align}
Replacing $||\mathbf{C}\mathbf{b}||_2^2$ in \eqref{D1} with the right hand side of \eqref{D3}, we immediately obtain $\nabla^2\Xi(\mathbf{b})\preceq L_{\Xi}(\phi)\,\mathbf{I}_K$.

\section{Proof of Theorem 4}

To prove part (i) of this theorem, notice that $\mathrm{Rank}\left(\mathbf{C}\right)=\mathrm{Rank}\left([\mathbf{g}_1,\cdots,\mathbf{g}_K]\right)$ due to the definition of $\mathbf{C}$ in \eqref{Cn}.
As $\mathrm{Rank}\left(\mathbf{C}^H\mathbf{C}\right)=\mathrm{Rank}\left(\mathbf{C}\right)$, $\mathrm{Rank}\left(\mathbf{C}^H\mathbf{C}\right)=K$ and $\lambda_{\mathrm{min}}\left(\mathbf{C}^H\mathbf{C}\right)>0$.
Putting $\lambda_{\mathrm{min}}\left(\mathbf{C}^H\mathbf{C}\right)>0$ and $\phi=0$ into \eqref{Lip1}, we obtain
\begin{align}
&L_{\Xi}(0)=\frac{2\sqrt{\beta}\big|\big|\mathrm{Re}\left(\mathbf{C}^H\mathbf{C}\right)\big|\big|_2}
{\sqrt{\lambda_{\mathrm{min}}\left(\mathbf{C}^H\mathbf{C}\right)/K}}<+\infty.
\end{align}

Next, to prove part (ii), it can be seen that $\lambda_{\mathrm{min}}\left(\mathbf{C}^H\mathbf{C}\right)=0$ if $\mathrm{Rank}\left([\mathbf{g}_1,\cdots,\mathbf{g}_K]\right)\neq K$.
Putting this result and $\phi=0$ into \eqref{Lip1}, we obtain $L_{\Xi}(0)=+\infty$.
On the other hand, if $\phi>0$, we must have $\sqrt{\phi^2+||\mathbf{C}\mathbf{b}||_2^2}>0$ due to $||\mathbf{C}\mathbf{b}||_2^2\geq 0$.
Putting this result into \eqref{Lip1}, we obtain $L_{\Xi}(\phi)<+\infty$ if $\phi>0$.

Finally, to prove part (iii) of this theorem, we need the following lemma.
\begin{lemma}
$\Xi(\mathbf{b})$ is bounded as
$\Phi(\mathbf{b})\leq \Xi(\mathbf{b})\leq \Phi(\mathbf{b})+2\sqrt{\beta}\,\phi$.
\end{lemma}
\begin{proof}
To prove the left inequality, notice that
$
||\mathbf{C}\mathbf{b}||_2
\leq
\sqrt{\phi^2+
||\mathbf{C}\mathbf{b}||_2^2}.
$
Putting this result into $\Xi(\mathbf{b})$ in \eqref{Xi}, we immediately obtain
$\Phi(\mathbf{b})\leq \Xi(\mathbf{b})$.
On the other hand, to prove the right inequality, we first compute
\begin{align}
&\Xi(\mathbf{b})-\Phi(\mathbf{b})
=2\sqrt{\beta}\Big(
\sqrt{\phi^2+
||\mathbf{C}\mathbf{b}||_2^2}
-\sqrt{||\mathbf{C}\mathbf{b}||_2^2}\Big). \label{G1}
\end{align}
Then applying the identity
$\sqrt{\phi^2+x}-\sqrt{x}=\frac{\phi^2}{\sqrt{\phi^2+x}+\sqrt{x}}
\leq \phi$, where the inequality is due to the monotonic decreasing feature of $\phi^2/(\sqrt{\phi^2+x}+\sqrt{x})$ with respect to $x$,
the right hand side of \eqref{G1} is upper bounded as
\begin{align}
&2\sqrt{\beta}\Big(
\sqrt{\phi^2+
||\mathbf{C}\mathbf{b}||_2^2}
-\sqrt{||\mathbf{C}\mathbf{b}||_2^2}\Big)
\leq 2\sqrt{\beta}\,\phi.
\end{align}
Putting this result into \eqref{G1}, we have $\Xi(\mathbf{b})-\Phi(\mathbf{b})\leq 2\sqrt{\beta}\, \phi$.
\end{proof}

Now, if an $\epsilon$-optimal solution $\mathbf{b}'\in\Delta$ to $\mathcal{Q}_2$ is obtained with $\Xi(\mathbf{b}')-\Xi(\mathbf{b}^\diamond)\leq \epsilon$, then we must have
\begin{align}
\Phi(\mathbf{b}')\leq\Xi(\mathbf{b}^\diamond)+\epsilon, \label{G2}
\end{align}
due to $\Phi(\mathbf{b})\leq\Xi(\mathbf{b})$ from the first inequality of Lemma 5.
On the other hand, taking the minimum on both sides of the second inequality of Lemma 5, we have
\begin{align}
&\mathop{\mathrm{min}}_{\mathbf{b}\in\Delta}~\Xi(\mathbf{b})\leq \mathop{\mathrm{min}}_{\mathbf{b}\in\Delta}~\Phi(\mathbf{b})+2\sqrt{\beta}\,\phi. \label{G3}
\end{align}
Putting $\mathop{\mathrm{min}}_{\mathbf{b}\in\Delta}~\Xi(\mathbf{b})=\Xi(\mathbf{b}^\diamond)$ and $\mathop{\mathrm{min}}_{\mathbf{b}\in\Delta}~\Phi(\mathbf{b})=\Phi(\mathbf{b}^*)$ into \eqref{G3}, \eqref{G3} becomes
$\Xi(\mathbf{b}^\diamond)\leq \Phi(\mathbf{b}^*)+2\sqrt{\beta}\,\phi$.
Combining this result with \eqref{G2} leads to $\Phi(\mathbf{b}')\leq \Phi(\mathbf{b}^*)+2\sqrt{\beta}\, \phi+\epsilon$, and part (iii) of this theorem is proved.

\end{document}